\documentclass[journal, onecolumn]{IEEEtran}

\usepackage[margin=3cm]{geometry}

\usepackage[T1]{fontenc}
\usepackage[english]{babel}
\usepackage{url,xspace}
\usepackage[usenames,dvipsnames]{color}
\usepackage{stmaryrd}
\SetSymbolFont{stmry}{bold}{U}{stmry}{m}{n}
\usepackage{paralist}

\usepackage{graphicx}
\graphicspath{{./plots/}{./results/}}

\usepackage{wrapfig}

\usepackage{authblk}

\usepackage{xspace}

\newcommand{\todo}[2][P]{\xspace\noindent{\color{blue}{[#1] #2}}\xspace}

\newcommand{\cut}[1]{}

\usepackage{fixltx2e}

\usepackage{array}

\usepackage{tikz}
\usetikzlibrary{calc,positioning}
\usetikzlibrary{shapes}
\usetikzlibrary{matrix}
\usetikzlibrary{decorations}
\usetikzlibrary{decorations.pathmorphing}
\usetikzlibrary{decorations.markings}
\usetikzlibrary{backgrounds}

\usepackage[colorlinks]{hyperref}
\usepackage{import}

\usepackage[linesnumbered,vlined,noend,ruled,nofillcomment]{algorithm2e}
\usepackage[justification=centering]{caption}

\usepackage{ifthen}
\usepackage{arrayjobx}

\usepackage{amsmath,amsthm,amssymb}
\usepackage{mathtools}
\DeclarePairedDelimiter{\ceil}{\lceil}{\rceil}

\usepackage[section]{placeins}

\usepackage{mathtools}

\DeclarePairedDelimiter\abs{\lvert}{\rvert}%
\DeclarePairedDelimiter\norm{\lVert}{\rVert}%

\makeatletter
\let\oldabs\abs
\def\abs{\@ifstar{\oldabs}{\oldabs*}}
\let\oldnorm\norm
\def\norm{\@ifstar{\oldnorm}{\oldnorm*}}
\makeatother

\newcommand{\ema}[1]{\ensuremath{#1}\xspace}

\newtheorem{theorem}{Theorem}

\newtheorem{lemma}{Lemma}

\theoremstyle{definition}

\def\vsls{0cm}
\def\vslm{0cm}
\def\vslb{0cm}

\newcommand{\vs}[1]{\vspace*{-#1}}
\newcommand{\vssm}{\vs{\vsls}}
\newcommand{\vsme}{\vs{\vslm}}
\newcommand{\vsmb}{\vs{\vslmb}}
\newcommand{\vsbi}{\vs{\vslb}}

\newcommand{\compsubsubsection}[1]{\subsubsection{#1}}

\newcommand{\ie}{\textit{i.e.}\xspace}

\newcommand{\eg}{\textit{e.g.}\xspace}

\newcommand{\inti}[2][1]{\ema{[#1..#2]}}
\newcommand{\clt}{Central Limit Theorem\xspace}
\newcommand{\step}[1]{\ema{\mathbf{1}_{#1}}}

\newcommand{\sdss}{search data structures\xspace}
\newcommand{\sds}{search data structure\xspace}
\newcommand{\Sdss}{Search data structures\xspace}

\newcommand{\sdstits}{Search Data Structures\xspace}

\newcommand{\ds}{data structure\xspace}
\newcommand{\dss}{data structures\xspace}

\newcommand{\shds}{structure\xspace}
\newcommand{\shdss}{structures\xspace}

\newcommand{\lf}{lock-free\xspace}
\newcommand{\Lf}{Lock-free\xspace}

\newcommand{\lftit}{Lock-Free\xspace}

\newcommand{\dsll}{linked list\xspace}
\newcommand{\dslls}{linked lists\xspace}

\newcommand{\dsLls}{Linked lists\xspace}
\newcommand{\dslltit}{Linked List\xspace}

\newcommand{\dssl}{skip list\xspace}
\newcommand{\dssls}{skip lists\xspace}
\newcommand{\dsSl}{Skip list\xspace}
\newcommand{\dsSls}{Skip lists\xspace}
\newcommand{\dssltit}{Skip List\xspace}

 \newcommand{\dsht}{hash table\xspace}
\newcommand{\dshts}{hash tables\xspace}

\newcommand{\dshttit}{Hash Table\xspace}

 \newcommand{\dsbt}{binary tree\xspace}
\newcommand{\dsbts}{binary trees\xspace}
 
\newcommand{\dsBts}{Binary trees\xspace}
\newcommand{\dsbttit}{Binary Tree\xspace}

\newcommand{\src}{\FuncSty{Search}\xspace}
\newcommand{\ins}{\FuncSty{Insert}\xspace}
\newcommand{\del}{\FuncSty{Delete}\xspace}
\newcommand{\nulltxt}{{\textit{null}}\xspace}

\newcommand{\castxtl}{\textit{Compare-and-Swap}\xspace}
\newcommand{\castxts}{\textit{CAS}\xspace}

\newcommand{\Read}{\textit{Read}\xspace}

\newcommand{\range}{\ema{\mathcal{R}}}
\newcommand{\nbn}{\ema{\mathcal{N}}}
\newcommand{\key}[1]{\ema{K_i}}

\def\fcas{\mathit{cas}}
\def\fread{\mathit{read}}
\def\ftrav{\mathit{trav}}
\def\fint{\mathit{int}}
\def\fext{\mathit{ext}}
\def\fmem{e}
\def\fcache{\mathit{cache}}
\def\fcmp{\mathit{cmp}}
\def\fapp{\mathit{app}}
\def\fdat{\mathit{dat}}
\def\ftlb{\mathit{tlb}}
\def\frec{\mathit{rec}}

\newcommand{\Cas}{\ema{cas}}

\newcommand{\formread}{\ema{\fread}}
\newcommand{\formcas}{\ema{\fcas}}
\newcommand{\formtrav}{\ema{\ftrav}}
\newcommand{\formmem}{\ema{\fmem}}
\newcommand{\formint}{\ema{\fint}}
\newcommand{\formext}{\ema{\fext}}

\newcommand{\formdat}{\ema{\fdat}}

\newcommand{\opcas}[1]{\ema{\fcas\!\left( #1 \right)}}
\newcommand{\opread}[1]{\ema{\fread\!\left( #1 \right)}}

\newcommand{\opmem}[1]{\ema{\formmem\!\left(#1\right)}}

\newcommand{\genlat}[1]{\ema{t^{#1}}}
\newcommand{\latdat}[1]{\ema{\genlat{\fdat}_{#1}}}
\newcommand{\lattlb}[1]{\ema{\genlat{\ftlb}_{#1}}}
\newcommand{\latcas}{\ema{\genlat{\fcas}}}
\newcommand{\latcmp}{\ema{\genlat{\fcmp}}}
\newcommand{\latapp}{\ema{\genlat{\fapp}}}
\newcommand{\latrec}{\ema{\genlat{\frec}}}

\newcommand{\genrate}[2]{\ema{\lambda_{#1}^{#2}}}
\newcommand{\rateCas}[1]{\ema{\genrate{#1}{\formcas}}}
\newcommand{\rateRead}[1]{\ema{\genrate{#1}{\formread}}}
\newcommand{\rateTra}[1]{\ema{\genrate{#1}{\formtrav}}}

\newcommand{\rateReadInt}[1]{\rateRead{\fint,#1}}
\newcommand{\rateCasInt}[1]{\rateCas{\fint,#1}}

\newcommand{\rateReadIE}[2]{\rateRead{#1,#2}}
\newcommand{\rateReadIntv}[2]{\rateRead{\fint,#1,#2}}
\newcommand{\rateCasEI}[2]{\rateCas{\fint,#1,#2}}
\newcommand{\rateCasEIA}[2]{\rateCas{#1,#2}}

\newcommand{\subt}[1]{\ema{\mathit{Sub}_{#1}}}

\newcommand{\proba}[1]{\ema{\mathbb{P}\left[ #1 \right]}}
\newcommand{\expec}[1]{\ema{\mathbb{E}\left[ #1 \right]}}
\newcommand{\DS}{\ema{D}}

\newcommand{\noded}[2]{\ema{N_{#1}^{#2}}}
\newcommand{\presprod}[2]{\ema{p_{#1}^{#2}}}
\newcommand{\nodet}[3]{\noded{#2,#3}{#1}}

\newcommand{\node}[1]{\noded{#1}{}}
\newcommand{\prespro}[1]{\presprod{#1}{}}

\newcommand{\formrou}{\mathit{rou}}
\newcommand{\nodero}[2]{\noded{#1,#2}{\formrou}}

\newcommand{\nodeda}[2]{\noded{#1,#2}{\formdat}}

\newcommand{\noderd}[3]{\noded{#2,#3}{\mathit{#1}}}
\newcommand{\presprord}[3]{\presprod{#2,#3}{\mathit{#1}}}

\newcommand{\nodeint}[1]{\noded{#1}{\formint}}
\newcommand{\presproint}[1]{\presprod{#1}{\formint}}

\newcommand{\nodeext}[1]{\noded{#1}{\formext}}
\newcommand{\presproext}[1]{\presprod{#1}{\formext}}

\newcommand{\presproie}[2]{\presprod{#1}{#2}}

\newcommand{\plin}[1]{\ema{q_{#1}}}

\newcommand{\nbvirt}[1]{\ema{H_{#1}}}
\newcommand{\nodeintv}[2]{\noded{#1,#2}{\formint}}
\newcommand{\presprointv}[2]{\presprod{#1,#2}{\formint}}

\newcommand{\pres}[1]{\ema{P_{#1}}}

\newcommand{\prs}[2]{\ema{Y_{#1}^{#2}}}
\newcommand{\nbpg}{\ema{\mathcal{M}}}

\newcommand{\keyv}[1]{\ema{K_{#1}}}

\newcommand{\hmax}{\ema{\mathit{h_{\mathit{max}}}}}

\newcommand{\loadf}{\ema{\mathit{lf}}}

\newcommand{\popu}[1]{\ema{s_{#1}}}

\newcommand{\op}[2]{\ema{op^{#1}_{#2}}}

\newcommand{\thr}{\ema{\mathcal{T}}}
\newcommand{\nbth}{\ema{P}}

\newcommand{\nbthS}[1]{\ema{P_{#1}}}
\newcommand{\latrecS}[1]{\ema{\genlat{\frec}_{#1}}}

\newcommand{\latrecSl}{\latrecS{\mathit{low}}}
\newcommand{\latrecSh}{\latrecS{\mathit{high}}}

\makeatletter
\newcommand{\removelatexerror}{\let\@latex@error\@gobble}
\makeatother

\newcommand{\abstalgo}{
\removelatexerror
\begin{procedure}[H]
\SetKwData{key}{key}
\SetKwData{op}{operation}
\SetKwData{kpdf}{keyPMF}
\SetKwData{opdf}{operationPMF}
\SetKwData{res}{result}
\SetKwData{pdo}{done}

\SetKwFunction{SDS}{SearchDataStructure}
\SetKwFunction{Skey}{SelectKey}
\SetKwFunction{Sop}{SelectOperation}

\SetAlgoLined
\While{! \pdo}{\nllabel{alg:li-bwl}
\key $\leftarrow$ \Skey{\kpdf}\;\nllabel{alg:li-ksel}
\op $\leftarrow$ \Sop{\opdf}\;\nllabel{alg:li-osel}
\res $\leftarrow$  \SDS{\key, \op}; \nllabel{alg:li-sds}
}
\caption{AbstractAlgorithm()\label{alg:gen-name}}
\end{procedure}
}

\newcommand{\trav}[1]{\ema{\mathit{Traverse}_{#1}}}
\newcommand{\casrv}[2]{\ema{\mathit{CAS}_{#1}^{\mathit{#2}}}}
\newcommand{\hitrv}[2]{\ema{\mathit{Hit}_{#1}^{\mathit{#2}}}}
\newcommand{\casexe}[1]{\casrv{#1}{exe}}
\newcommand{\cassta}[1]{\casrv{#1}{stall}}
\newcommand{\casrec}[1]{\casrv{#1}{reco}}
\newcommand{\hitcache}[2]{\hitrv{#1}{cache_{#2}}}

\newcommand{\hittlb}[2]{\hitrv{#1}{tlb_{#2}}}

\newcommand{\hitsim}[1]{\hitrv{#1}{}}

\newcommand{\obj}[1]{\ema{O_{#1}}}

\newcommand{\oprv}{\ema{\mathit{Op}}}

\newcommand{\formins}{\ema{\mathit{ins}}}
\newcommand{\formdel}{\ema{\mathit{del}}}
\newcommand{\formsrc}{\ema{\mathit{src}}}

\newcommand{\chatio}[2]{\ema{T_{#1}^{#2}}}
\newcommand{\chati}[1]{\ema{T_{#1}^{\fdat}}}
\newcommand{\chatit}[1]{\ema{T_{#1}^{\ftlb}}}
\newcommand{\casi}[1]{\ema{C_{#1}^{\fdat}}}
\newcommand{\casit}[1]{\ema{C_{#1}^{\ftlb}}}

\newcommand{\xcaf}[2]{\ema{X^{#1}(#2)}}
\newcommand{\xcas}[1]{\ema{X\!(#1)}}

\newcommand{\hito}[2]{\ema{\mathit{hit}_{#1}^{#2}}}

\newcommand\rr[1]{#1}

\newcommand\pp[1]{}

\newcommand\pr[2]{#2}

\title{\lftit \sdstits: Throughput Modelling with Poisson Processes}

\pr{

\authorrunning{A. Atalar, P. Renaud-Goud and P. Tsigas} 

\Copyright{Aras Atalar, Paul Renaud-Goud and Philippas Tsigas}

\subjclass{D.1.3 Concurrent Programming}
\keywords{Lock-free, Data Structures, Performance, Modeling, Analysis}

\author{Aras Atalar}{Chalmers University of Technology}{aaras@chalmers.se}{}{}
\author{Paul Renaud-Goud}{Informatics Research Institute of Toulouse}{Paul.Renaud.Goud@irit.fr}{}{}
\author{Philippas Tsigas}{Chalmers University of Technology}{philippas.tsigas@chalmers.se}{}{}
}
{

\author[1]{Aras Atalar}
\author[2]{Paul Renaud-Goud}
\author[1]{Philippas Tsigas}

\affil[1]{Chalmers University of Technology\\
  \texttt{aaras|philippas.tsigas@chalmers.se}}
\affil[2]{Informatics Research Institute of Toulouse\\
  \texttt{Paul.Renaud.Goud@irit.fr}}

}

\begin{document}

\SetFuncSty{textsf}


\maketitle


\begin{abstract}
\pp{{\footnotesize\bf (Eligible for best student paper award)\\}}
This paper considers the modelling and the analysis of the performance
of \lf concurrent \sdss.  Our analysis considers such
\lf \dss that are utilized through a sequence of operations
which are generated with a memoryless and stationary access pattern.
Our main contribution is a new way of analysing \lf \sdss: 
our execution model matches with the behavior 
that we observe in practice and achieves good throughput predictions.
\Sdss are formed of linked basic blocks, usually
referred as nodes, that can be accessed by two kinds of events,
characterized by their latencies; (i) \castxts events originated as a
result of modifications of the \sdss (ii) \Read events originated during traversals.
This type of \dss are usually designed to accommodate a large number of
data nodes, which makes the occurrence of an event on a given node rare 
at any given time.
The throughput is defined by the number of events per operation 
in conjunction with the factors that impact the latencies of these 
events. We frame these impacting factors under capacity and 
coherence cache misses.

In this context, we model the events as Poisson processes that we can merge
and split to estimate the latencies of the events based on the
interleaving of events from different threads, and in turn estimate
the throughput. We have validated our analysis on several 
fundamental \lf \sdss such as \dslls, \dshts, \dssls and \dsbts.

\end{abstract}
\rr{\newpage}

\rr{\tableofcontents\newpage}

\section{Introduction}

A \sds is a collection of $\langle\mathit{key}, \mathit{value}\rangle$
pairs which are stored in an organized way to allow efficient search,
delete and insert operations. \dsLls, \dshts, \dsbts are some widely
known examples. \Lf implementations of such concurrent \dss are known
to be strongly competitive at tackling scalability by
allowing processors to operate asynchronously on the \ds.

Performance (here throughput, \ie number of operations per unit of
time) is ruled by the number of events in a
\sds operation 
(\eg $O(\log \nbn)$ for the expected
number of steps in a \dssl or a \dsbt). The practical performance
estimation requires an additional layer as the cost (latency) of these
events need to be mapped onto the hardware platform; typical values of
latency varies from 4 cycles for an access to the
first level of cache, to 350 cycles for the last level
of remote cache. 
To estimate the latency of events, one needs 
to consider the misses, which are sensitive to
the interleaving of these events on the
time line. 
On the one hand, a capacity miss in data or TLB (Translation Lookaside Buffer) 
caches with LRU (Least Recently Used) policy
arise when the interleaving of memory accesses evicted a cacheline.
On the other hand, the coherence cache misses arise as a result of the modifications, 
that are often realized with \castxtl (\castxts) instructions, in the \lf \sds.
The interleaving of events that originate
from different threads, determine the 
frequency and severity of these misses, hence
the latencies of the events.

In the literature, there exist many asymptotic analyses on the time
complexity of sequential \sdss and amortized analyses for the
concurrent \lf variants that involve the interaction between multiple
threads. But they only consider the number of events, ignoring the
latency.  On the other side, there are performance analyses that aim
to estimate the coherence and capacity misses for the programs on a
given platform, with no view on \dss.
We will mention them in
the related work. However, there is a lack of results that merge these
approaches in the context of \lf \dss to analytically predict the
practical performance.

An analytical performance prediction framework could be useful in many
ways: (i) to facilitate design decisions by providing an extensive
understanding; (ii) to rank different designs in various contexts; (iii)
to help the tuning process. On this last point, \lf \dss come with
specific parameters, \eg padding, back-off and memory management
related parameters, and become competitive only after picking their
hopefully optimal values.

In this paper, we aim to compute the {\it average} throughput of \sdss 
for a sequence of operations, generated by a memoryless and stationary
access pattern. The threads execute the same piece of code on
the same platform, throughput \thr can be estimated on the long-term
as the expected latency of an operation (subjected to the distribution
of the operations) divided by the number of threads \nbth. As the
traversal of a \sds is light in computation, the latency of an
operation is dominated by the memory access costs to the nodes that
belong to the path from the entry of the \ds to the targeted node.

Therefore, part of this paper is dedicated to the discovery of the
route(s) followed by a thread on its way to reach any node in the \ds.
In other words, what is the sequence of nodes that are accessed when a
given node is targeted by an operation.

As the latency of an operation is the sum of the latency of each
memory access to the nodes that are on the path, we obviously need to
estimate the individual latency of each traversed node. Even if, in
the end, we are interested in the average throughput, this part of the
analysis cannot be satisfied with a high-level approach, where we
would ignore which thread accesses which node across time. For
instance, the cache, whose misses are expected to greatly impact
throughput, should be taken carefully into account. This can only be
done in a framework from which the interleaving of memory accesses among
threads can be extracted\cut{ (more details follow, in
Section~\ref{sec:factors})}. That is why we model the distribution of
the memory accesses for every thread.

More precisely, a memory access ({\it traversal}) can be either the
read or the modification of a node, and two point distributions per
node represent the triggering instant of either a \Read or a
\castxts. These point distributions are modelled as Poisson point
processes, since they can be approximated by Bernoulli processes, in
the context of rare events. Knowing the probabilistic ordering of
these events gives a decisive information that is used in the estimate
of the traversal latency associated with the triggered event. Once
this information is grabbed, we roll back to the expectation of the
traversal of a node, then to the expectation of the latency of an
operation.

We validate our approach through a large set of experiments on several
\lf \sdss based on various algorithmic designs, namely \dslls, \dshts,
\dssls and \dsbts. We feed our experiments with different key
distributions, and show that our framework is able to predict and
explain the observed phenomena.

The rest of the paper is organized as follows. We discuss related work
in Section~\ref{sec:related}, then the problem is formulated in
Section~\ref{sec:problem}. We present the framework in
Section~\ref{sec:framework} and the computation of throughput in
Section~\ref{sec:thput}. In Section~\ref{sec:dss}, we show how to
initiate our model by considering the particularity of different
\sdss. Finally, we describe the experimental results in
Sections~\ref{sec:expMain} and~\ref{sec:padMain}.\vssm

\section{Related Work}
\label{sec:related}

The search path length of skiplists is analysed
in~\cite{skiplist1,skiplist2}.  In~\cite{skiplist1}, the search path
length is split into vertical and horizontal components, where the
horizontal cost is modelled with the number of right-to-left maximas
(which corresponds to the traversed node) in a sequence of nodes with
random heights.  In~\cite{tree2,tree,tree3},
various performance shapers for the randomized trees are studied, such
as the time complexity of operations, the expectation and
distribution of the depth of the nodes based on their keys.


Previously mentioned studies are not concerned with the interaction
between the algorithms and the hardware. The following approaches rely on
the independent reference model (IRM) for memory references and derive
theoretical results or performance analysis.
In~\cite{cachemiss}, data reuse distance patterns are modelled and
then exploited to predict the cache miss ratio.
In~\cite{flajolet}, the
exact cache miss ratio is derived analytically (computationally
expensive) for LRU caches under IRM.
As an outcome of this approach, the cache miss ratio of a static
binary tree is estimated by assigning independent reference probababilities to the
nodes in~\cite{fix}.

For the time complexity of lock-free \sdss, asymptotic amortized analyses
\cite{fomitchev,bapi} are conducted since it is not possible to bound the
execution time of a single operation, by definition. Apart from these theoretical studies,
the performance of concurrent lock-free \sdss are studied and investigated through
empirical studies in~\cite{gramoli,trigonakis}.
In~\cite{davidfail}, it is shown experimentally that the conflicts
between threads occur very rarely in the context of concurrent \sdss,
which is confirmed by our analysis.\vssm

\section{Problem Statement}
\label{sec:problem}

We describe in this section the structure of the algorithm and the
system that is covered by our model. We target a multicore platform
where the communication between threads takes place through
asynchronous shared memory accesses. The threads are pinned to
separate cores and call \ref{alg:gen-name} (see
Figure~\ref{algo:abst}) when they are spawned.


\begin{wrapfigure}[9]{L}{0.44\textwidth}
\footnotesize
\abstalgo
\caption{Generic framework\label{algo:abst}}
\end{wrapfigure}


A concurrent \sds is a shared collection of data elements, each
associated with a key, that support three basic operations holding a
key as a parameter. \src (resp. \ins, \del) operation returns
(resp. inserts, deletes) the element if the
associated key is present (resp. absent, present) in the \sds,
otherwise returns \nulltxt.

The applications that use a \sds can be seen as a
sequence of operations on the \shds, interleaved by
application-specific code containing at least the key and operation
selection, as reflected in~\ref{alg:gen-name}.

The access pattern (\ie the output of the key and operation
selections) should be considered with care since it plays a decisive
role in the throughput value. An application that always
looks for the first element of a \dsll will obviously lead to very
high throughput rates.
%
%
In this study, we consider a memoryless and stationary key and
operation selection process \ie such that the probability of selecting
a key (resp. an operation type) is a constant.

A \sds is modelled as a set of basic blocks
called nodes, which either contain a value ({\it valued nodes}) or
routes towards nodes ({\it router nodes}).
W.l.o.g. the key set can be reduced to
\inti{\range}, where \range is the number of possible keys. We denote
by $\left(\node{i}\right)_{i \in \inti{\nbn}}$ the set of \nbn potential nodes,
and by \keyv{i} the key associated with \node{i}.
Until further notice (see Section~\ref{sec:padMain}), we assume that
we have exactly one node per cacheline.

An operation can trigger two types of events in a
node. We distinguish these events as \Read and \castxts events. The
latency of an event is based on the state of the hardware platform at
the time that the event occurs, \eg the level of the cache where a node
belongs to for a \Read request. We summarize the parameters of our model
as follows:

\begin{itemize}
\item \textit{Algorithm parameters:} Expected latency of the
  application call \latapp, expected computational cost to traverse a
  node \latcmp, probability mass functions for the key and operation
  selection.
\item \textit{Platform parameters:} Cache hit latencies
  (resp. capacity) from level $\ell$: \latdat{\ell}
  (resp. \casi{\ell}) for the data caches and \lattlb{\ell}
  (resp. \casit{\ell}) for TLB caches;
  other memory instruction latencies (that depends on \nbth): \latcas for a \castxts execution and
  \latrec to recover from an invalid state; number of threads \nbth.
\end{itemize}

\cut{The aim of this study is to estimate the throughput \thr
of~\ref{alg:gen-name}, \ie the number of completed operations per unit
of time, from this set of parameters.}

\section{Framework}
\label{sec:framework}

\vsme
\subsection{Event Distributions}
\label{sec:distrib}

We consider first a single thread running~\ref{alg:gen-name} on a \ds
where only search operations happen, and we observe the distribution
of the \Read triggering events on a given node \node{i}. The execution
is composed of a sequence of search operations, where each operation
is associated with a set of traversed nodes, which potentially
includes \node{i}. If we slice the time into consecutive intervals,
where an interval begins with a call to an operation, we can model the
\Read events as a Bernoulli process (where a success means that a
\Read event on \node{i} occurs), where the probability of having a
\Read event during an interval depends on the associated operation
(recall that the operation generating process is stationary and
memoryless).

\Sdss have been designed as a way to store large data sets while still
being able to reach any node within a short time: the set of traversed
nodes is then expected to be small in front of the set of all
nodes. This implies that, given an operation, the probability that
\node{i} belongs to the set of traversed nodes is small.
Therefore we can map the Bernoulli process on the
timeline with constant-sized interval of length $\thr^{-1}$ instead of
mapping it with the actual operation intervals: as the probability of
having a \Read event within an operation is small, the duration
between two events is big, and this duration is close to the number of
initial intervals within this duration, multiplied by $\thr^{-1}$
(with high probability, because of the \clt).

When we increase the scope of the operations to insertion and
deletion, the \shds is no longer static and the probability for a node
to appear in an interval is no longer uniform, since it can move
inside the \ds. There exists a long line of research in approximating
Bernoulli processes by Poisson point processes
\cite{poisson1,poisson2,poisson3}. In particular,
\cite{poisson-nonuni} has dealt with non-uniform Bernoulli
processes.
Their error bounds, which are proportional to the success probability,
strengthen the use of Poisson processes in our context: the events on
\node{i} are rare, thus the probabilities in Bernoulli processes are
small and the approximation is well-conditioned.




%
Once the \Read and \castxts triggering events are modelled as Poisson
processes for a single thread, the merge of several Poisson processes
models the multi-thread execution.

Lastly, we specify a point on the dynamicity: since we have insertions
and deletions, nodes can enter and leave the \ds. This is modelled by
the masking random variable \pres{i} which expresses the presence of
\node{i} in the \shds. At a random time, we denote by \DS the set of
nodes that are inside the \ds, and \pres{i} is set to 1 iff $\node{i}
\in \DS$. We denote by \prespro{i} its probability of success
($\prespro{i}=\proba{\pres{i}=1}$).
Its evaluation will often rely on the probability that the last update
operation on key $k$ was an \ins; we denote it by \plin{k},
and
\pr{$\plin{k} = \proba{\oprv=\op{\formins}{k}}/(\proba{\oprv=\op{\formins}{k}}+ \proba{\oprv=\op{\formdel}{k}})$.}
{\[ \plin{k} = \frac{\proba{\oprv=\op{\formins}{k}}}{\proba{\oprv=\op{\formins}{k}} + \proba{\oprv=\op{\formdel}{k}}}.\]}
Note that the \sdss contain generally several {\it sentinel nodes}
which define the boundaries of the \shds and are never removed from
the \shds: their presence probability is $1$.

%
%

For a given node \node{i}, we denote by \rateTra{i}
(resp. \rateRead{i}, \rateCas{i}) the rate of the events triggering a
traversal (resp. \Read, \castxts) of \node{i} due to one thread, when
$\node{i} \in \DS$. \op{\formdel}{k} (resp. \op{\formins}{k},
\op{\formsrc}{k}) stands for a \del (resp. \ins, \src) on node
key $k$. The probability for the application to select \op{o}{k},
where $o \in \{\formins, \formdel, \formsrc\}$ is denoted by
\proba{\oprv=\op{o}{k}}. $\op{o}{k} \leadsto \formcas(\node{i})$
(resp. \formread(\node{i})) means that during the execution of \op{o}{k}, a \castxts
(resp. a \Read) occurs on \node{i}. Putting all together, we derive
the rate of the triggering events:

\vspace{-.3cm}\begin{equation}
  \label{eq:lambdas}
  \forall e \in \{ \formcas, \formread \} \; : \;
  \genrate{i}{e}=\frac{\thr}{\nbth} \times \sum_{o \in \{\formins, \formdel, \formsrc\}} \sum_{k=1}^{\range} \proba{\oprv=\op{o}{k}} \times \proba{ \op{o}{k} \leadsto e(\node{i}) \, |\, \node{i}\in\DS}\\
\end{equation}

Recall for later that Poisson processes have useful properties, \eg
merging two Poisson processes produces another Poisson process whose
rate is the sum of the two initial rates. This implies especially that
the traversal triggering events follows a Poisson process with rate
$\rateTra{i} = \rateRead{i} + \rateCas{i}$, and that the read
triggering events that originates from $\nbth'$ different threads and
occurs at \node{i} follow a Poisson process with rate $\nbth' \times
\rateRead{i}$.









\vsme
\subsection{Validity of Poisson Process Hypothesis}
\label{sec:valid}

\pr{

\begin{wrapfigure}[10]{l}{.42\textwidth}
\vspace{-.35cm}
\includegraphics[width=.42\textwidth]{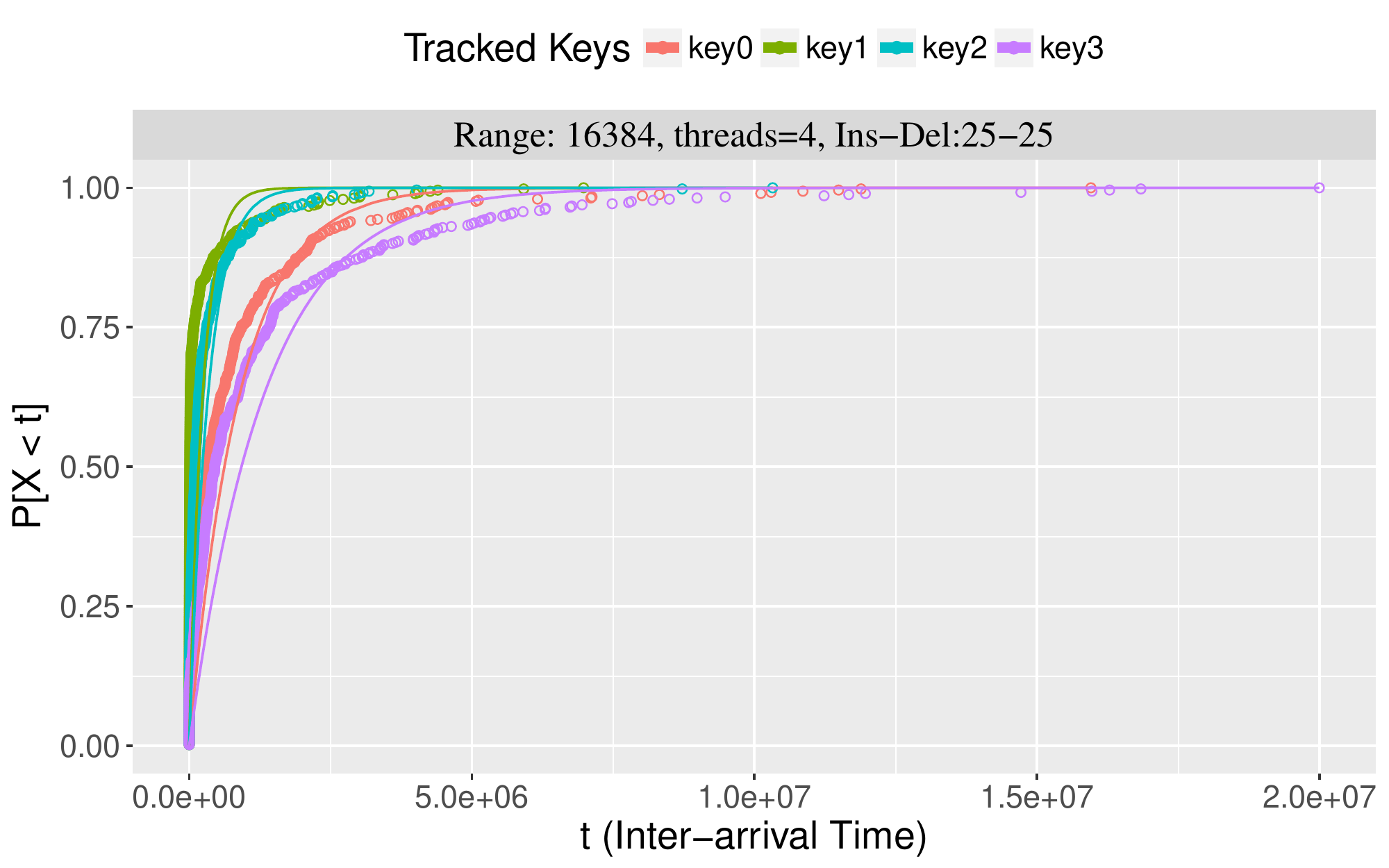}
\vspace{-.55cm}
\caption{Read Events for Skiplist}
\label{fig:poisModelSik}
\end{wrapfigure}

}
{

\begin{figure}[h!]
\begin{center}
\begin{minipage}{.47\textwidth}\begin{center}
\includegraphics[width=\textwidth]{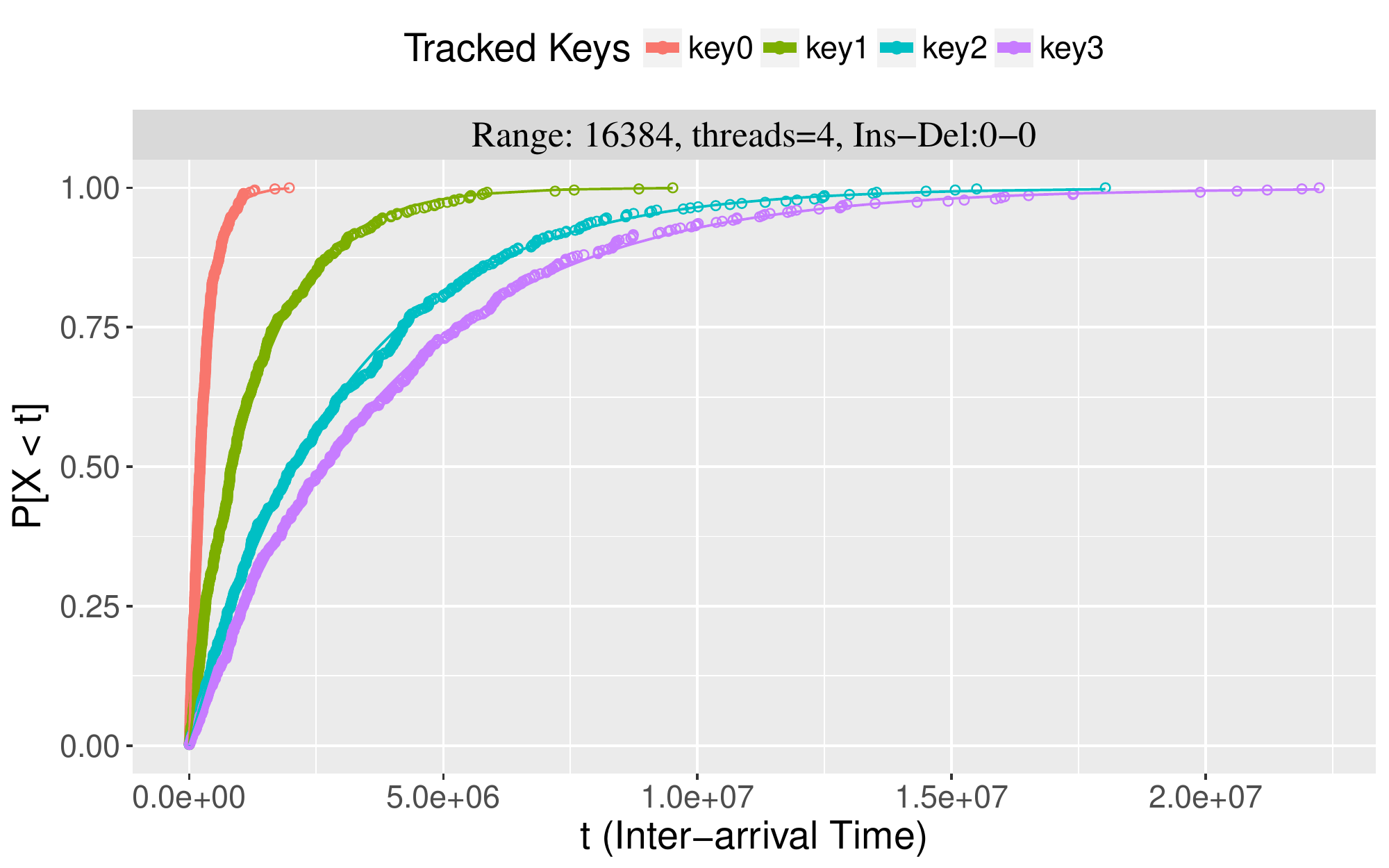}

(a) \Read Events for Skiplist
\end{center}\end{minipage}\hfill%
\begin{minipage}{.47\textwidth}\begin{center}
\includegraphics[width=\textwidth]{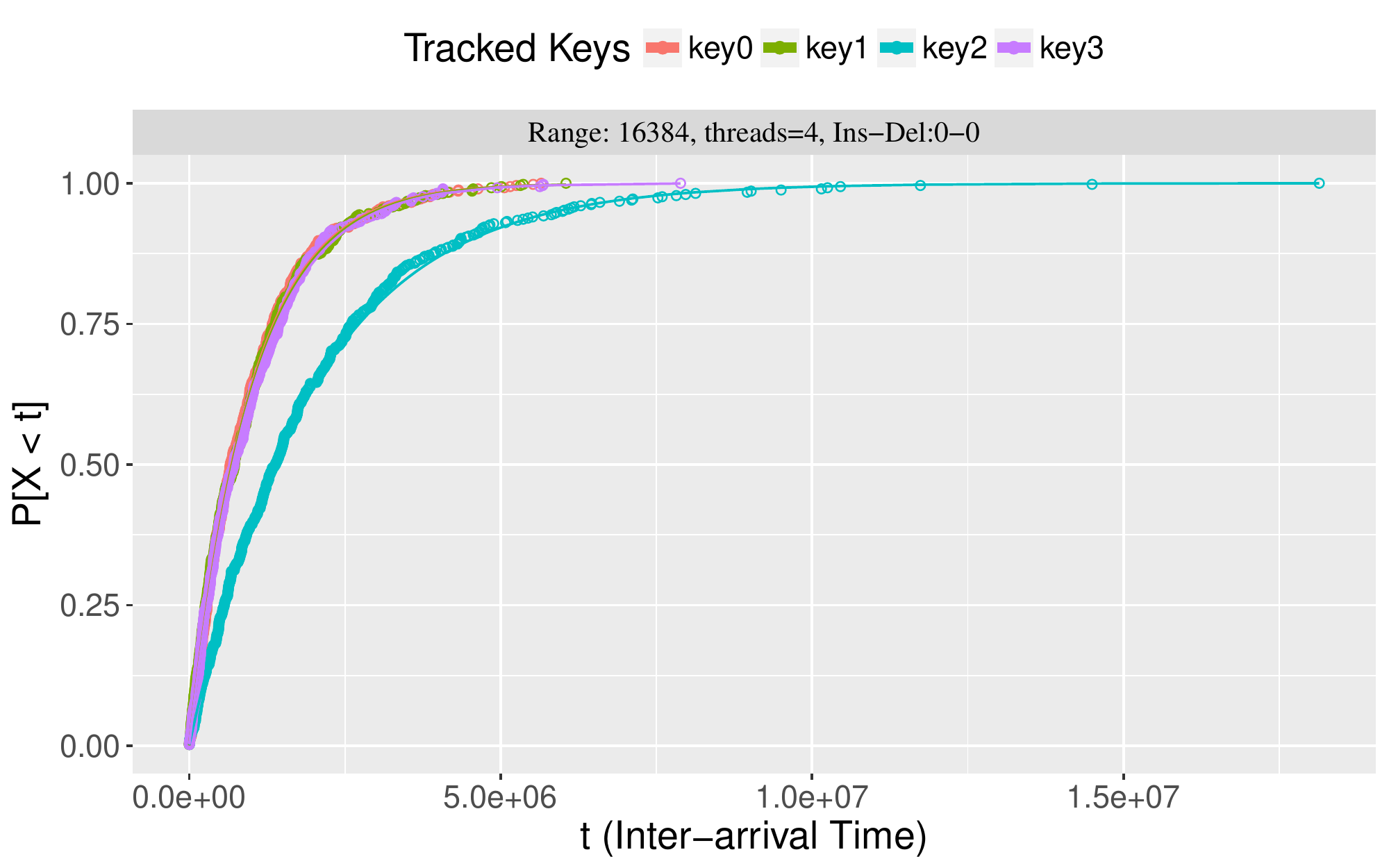}

(b) \Read Events for Hash Table
\end{center}\end{minipage}

\begin{minipage}{.47\textwidth}\begin{center}
\includegraphics[width=\textwidth]{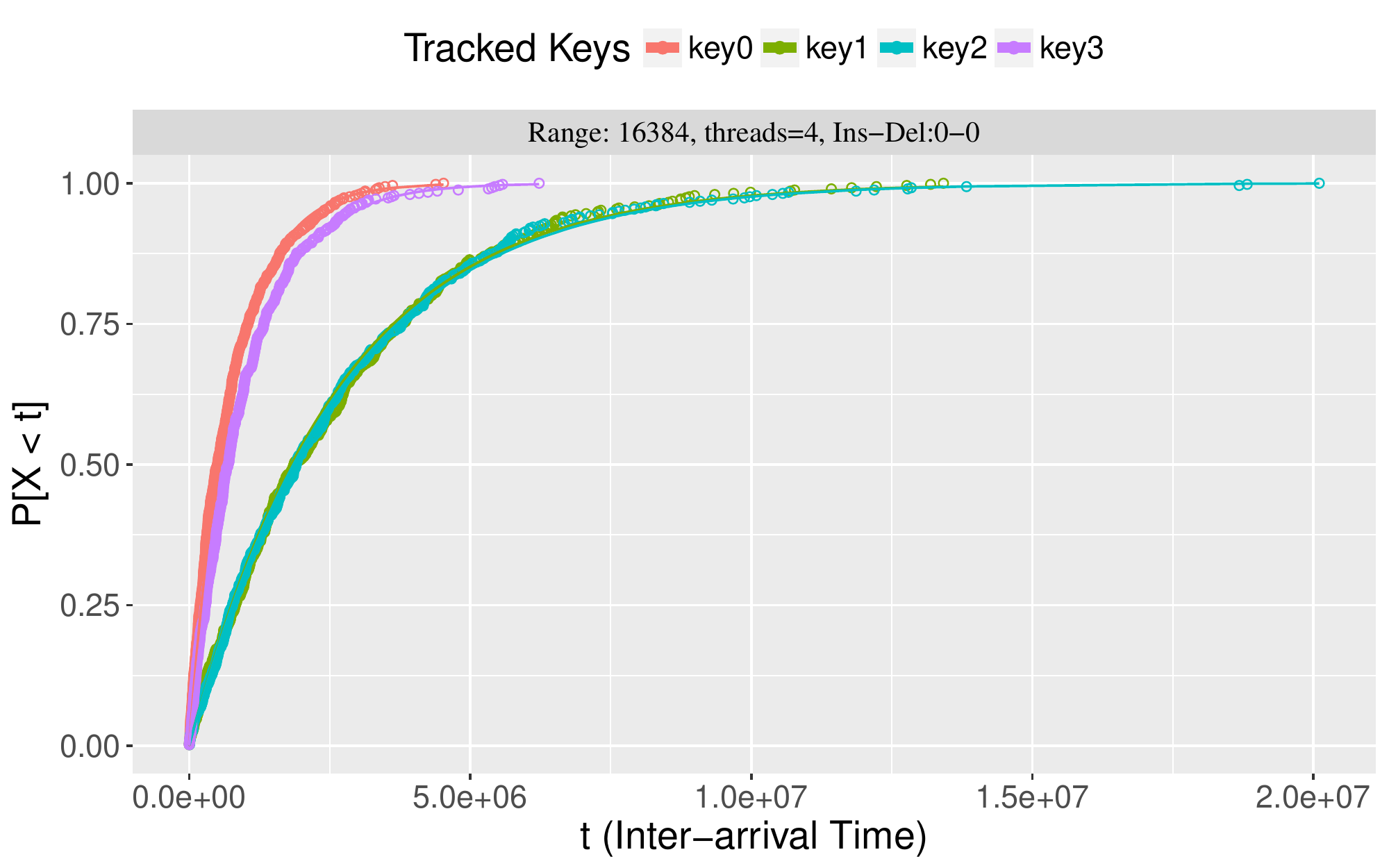}

(c) \Read Events for Binary Tree
\end{center}\end{minipage}\hfill%
\begin{minipage}{.47\textwidth}\begin{center}
\includegraphics[width=\textwidth]{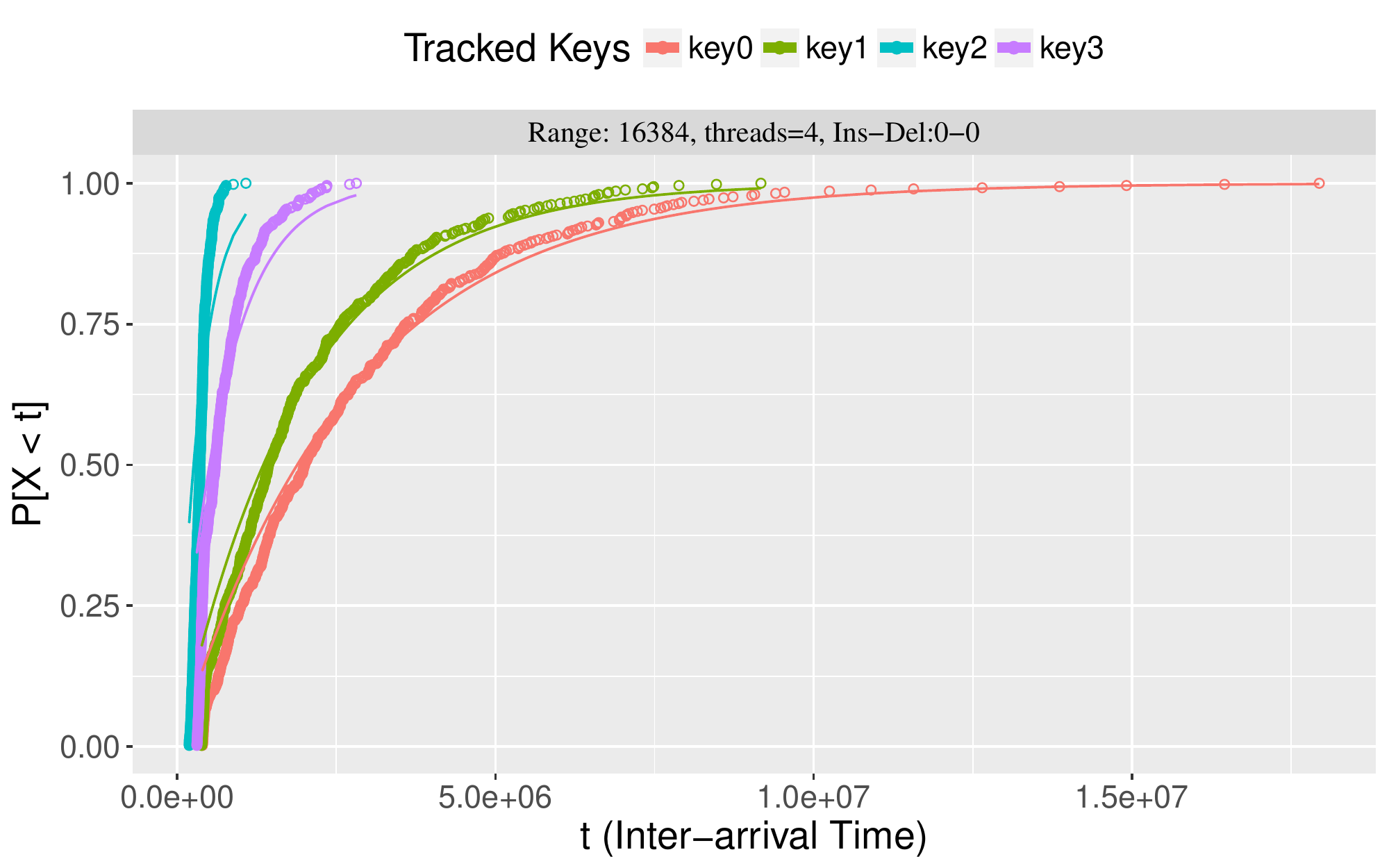}

(d) \Read Events for Linked List
\end{center}\end{minipage}

\end{center}
\caption{Poisson Process Modeling - Search Only\label{fig:poisModelRR}}
\end{figure}

\begin{figure}[h!]
\begin{center}
\begin{minipage}{.47\textwidth}\begin{center}
\includegraphics[width=\textwidth]{./results/pdfs/sklsik.pdf}

(a) \Read Events for Skiplist
\end{center}\end{minipage}\hfill%
\begin{minipage}{.47\textwidth}\begin{center}
\includegraphics[width=\textwidth]{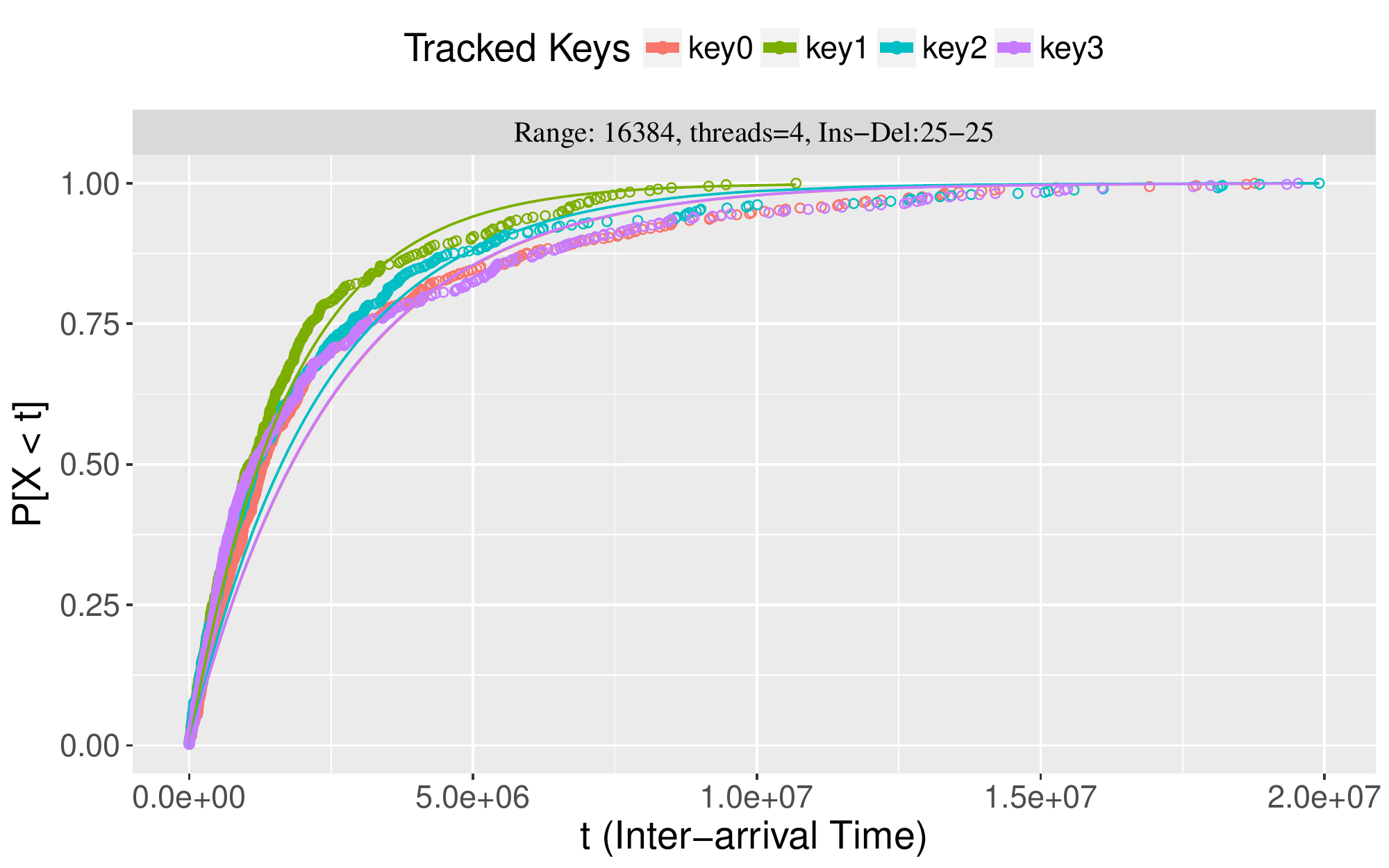}

(b) \Read Events for Hash Table
\end{center}\end{minipage}

\begin{minipage}{.47\textwidth}\begin{center}
\includegraphics[width=\textwidth]{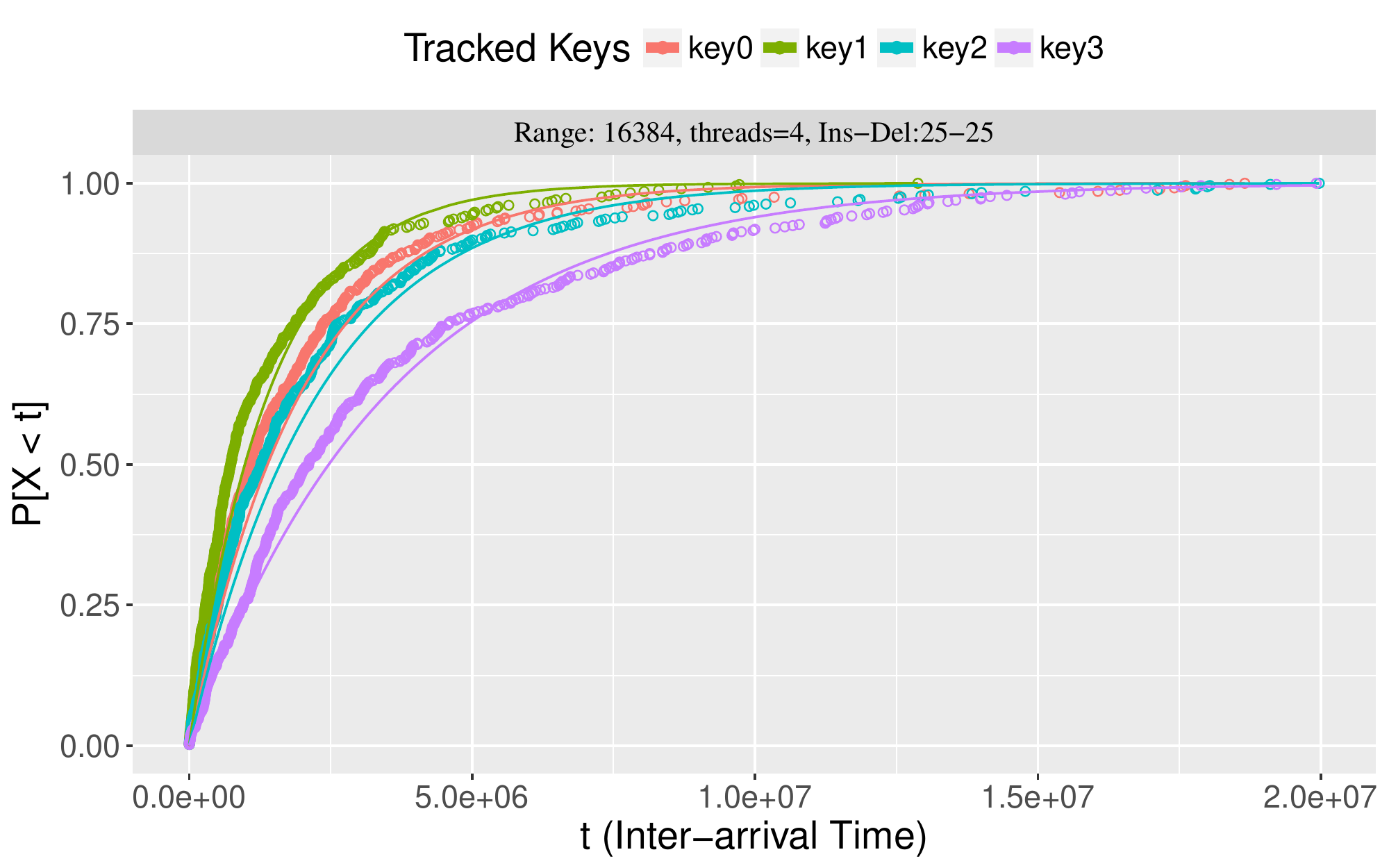}

(c) \Read Events for Binary Tree
\end{center}\end{minipage}\hfill%
\begin{minipage}{.47\textwidth}\begin{center}
\includegraphics[width=\textwidth]{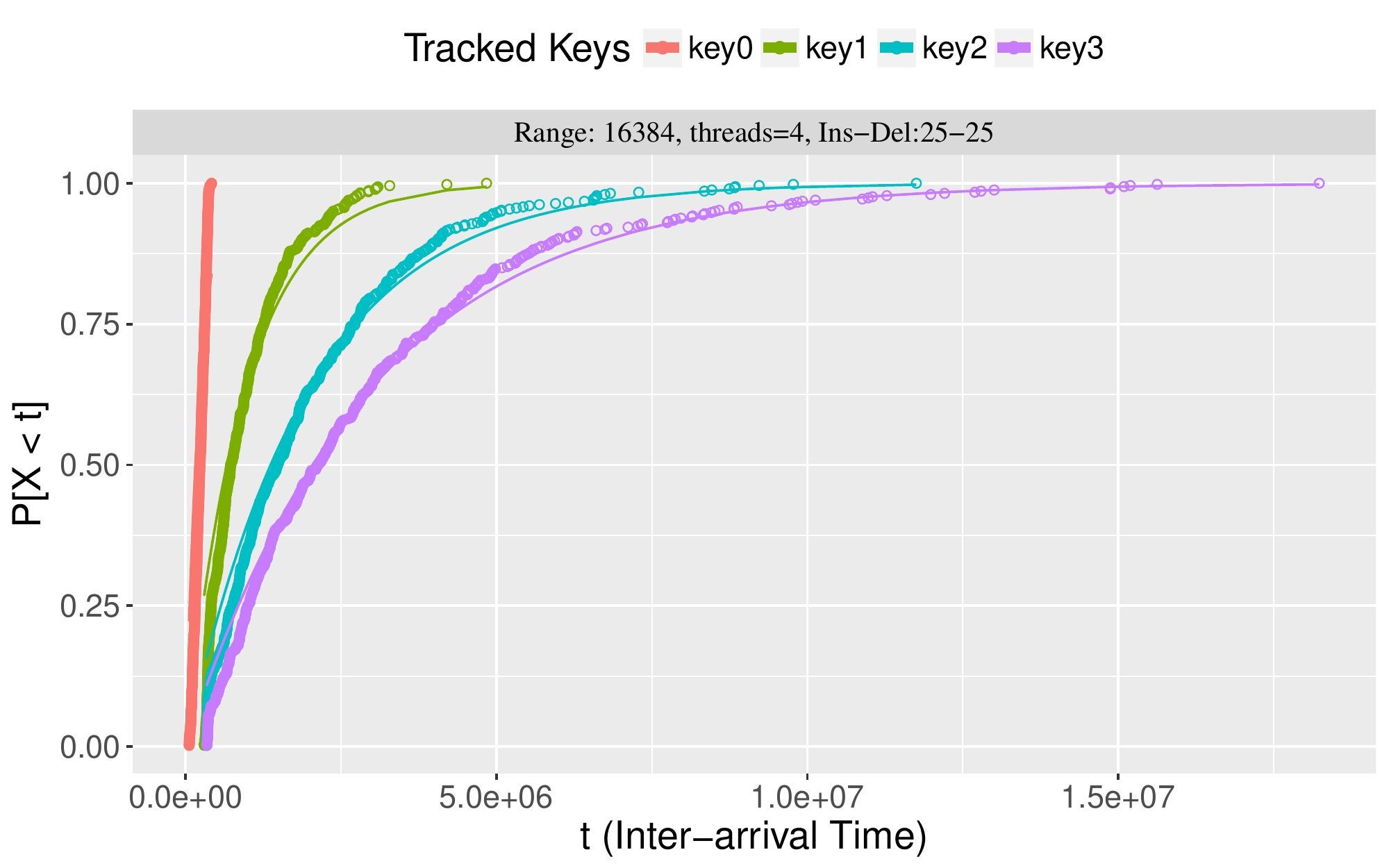}

(d) \Read Events for Linked List
\end{center}\end{minipage}

\end{center}
\caption{Poisson Process Modeling - 50/50 Search/Update\label{fig:poisModelRRSik}}
\end{figure}

}

To illustrate the validity of modeling the events as Poisson
processes, we experimentally extract the cumulative distribution
function of the inter-arrival latency of \Read events that occur on a
given node in a \dssl and we compare it against the corresponding exponential
distribution (recall that the time between events in a Poisson process
is exponentially distributed).

\pr{We consider here a scenario with $50/50$ search/update.}
{We consider a search only scenario and $50/50$ search/update scenario.}
Each thread initially picks a random key and tracks the instants when
a node associated with the chosen key is traversed during the
execution. To facilitate the recording of the inter-arrival times, we
disable the deletion of these particular keys (deletion is still
enabled for any other key).

\pr{In Figure~\ref{fig:poisModelSik}}{In
Figure~\ref{fig:poisModelRR}   and Figure~\ref{fig:poisModelRRSik}},
we illustrate the results, where the dots represent the experimental
measurements and the lines are generated by exponential
distributions. The mean of each distribution is instantiated as the
mean of the experimental measurements.
One can observe the grounds {\it a posteriori} of our Poisson process
modeling, and the variation of the event rates across keys, issuing
from the differences between the node characteristics (key, height,
location; see Section~\ref{sec:dss}).
\pp{A more comprehensive set of comparisons (other \dss and a search
  only scenario) can be found in~\cite{long-version}.}

\subsection{\vssm{}Impacting Factors}
\label{sec:factors}

We have identified five factors that dominate the traversal
latency of a node, distributed into two sets.
On the one hand, the first set of factors only emerges in the parallel
executions as a result of the coherence issues on the \sdss.
Atomic
primitives, such as a \castxts, are used to modify the shared \sdss
asynchronously.
%
%
To execute a \castxts in multi-core architectures, the cache coherency
protocol
enforces exclusive ownership of the target cacheline by
a thread (pinned to a core) through the invalidation of all the other
copies of the cacheline in the system, if needed. One can guess the
performance implications of this process that triggers back and forth
communication among the cores. As the first factor, \castxts
instruction has a significant latency. The thread that executes the
\castxts
pays this latency cost. Secondly, any other thread has to stall until
the end of the \castxts execution if it attempts to access (read or
modify) the node while the \castxts is getting executed.
Last and most importantly, any thread pays a cost to bring a cacheline
to a valid state if it attempts to access a node that resides in this
cacheline and that has been modified by another thread after its
previous access to this node.

On the other hand, the capacity misses in the data and TLB caches are
other performance impacting factors for the node traversals. Consider
a cache of size $C$ (fully associative), assume a node is traversed by
a thread at time $t$ and the next traversal (same thread and node)
occurs at time $t'$. The thread would experience a capacity miss for
the traversal at time $t'$ if it has traversed at least $C$ distinct
nodes in the interval $(t, t')$.  The same applies for TLB caches
where the references to the distinct pages are counted instead of the
nodes.

At a given instant, we denote by \trav{i} the latency of traversing
node \node{i}, either due to a \Read event or a \castxts event, for a
given thread. This latency is the sum of random variables that 
correspond to the previous respective five impacting factors:

\vspace{-.2cm}
\begin{equation} \label{eq:traverse}
  \trav{i} = \casexe{i} + \cassta{i} + \casrec{i} + \sum_{\ell} \hitcache{i}{\ell} + \sum_{\ell}\hittlb{i}{\ell},\vsme%
\end{equation}
where, at a random time, \casexe{i} is
the latency of a \castxts,
\cassta{i} 
the stall time implied by other threads executing a
\castxts on \node{i}, \casrec{i} 
the time needed to fetch the data
from another modifying thread, \hitcache{i}{\ell} 
the latency
resulting from a hit on the data cache in level $\ell$, and
\hittlb{i}{\ell} 
the latency coming from a hit on the TLB cache in
level $\ell$.\vssm

\vsme\subsection{Solving Process}
\label{sec:solving}

\vssm%
The solving decomposes into three main steps. Firstly, we can notice
that Equation~\ref{eq:lambdas} exposes $2\range + 1$ unknowns (the
$2\range$ access rates and throughput) against $2\range$ equations.
To end up with a unique solution, a last equation is necessary. The
first two steps provide a last sufficient equation thanks to Little's
law (see Section~\ref{sec:little}), which links throughput with the
expectation of the traversal latency of a node, computed from
Sections~\ref{sec:casexe} to~\ref{sec:pagemiss}. We show in these
sections that they can be expressed according to the access rates
\rateRead{i} and \rateCas{i}.
The last step focuses on the values of the probabilities in
Equation~\ref{eq:lambdas}, which are strongly related with the
particular \ds under consideration; they are instantiated in
Section~\ref{sec:ll} (resp.~\ref{sec:ht}, \ref{sec:sl}, \ref{sec:bt})
for \dslls (resp. \dshts, \dssls, \dsbts).\pr{\vsme}{}

\section{Throughput Estimation}
\label{sec:thput}

\subsection{Traversal Latency}

Applying expectation to Equation~\ref{eq:traverse} leads to
$\expec{\trav{i}} = \expec{\casexe{i}} + \expec{\cassta{i}} +
\expec{\casrec{i}} + \expec{\sum_{\ell} \hitcache{i}{\ell}} +
\expec{\sum_{\ell} \hittlb{i}{\ell}}$. We express here each term
according to the rates at every node \rateCas{\star} and
\rateRead{\star}.  


\compsubsubsection{CAS Execution}
\label{sec:casexe}

Naturally, among all traversal events, only the events originating
from a \castxts event contribute, with the latency \latcas of a
\castxts: $\expec{\casexe{i}} = \latcas \cdot
\rateCas{i}/(\rateRead{i}+\rateCas{i})$.




\compsubsubsection{Stall Time}
\label{sec:stall}

A thread experiences stall time while traversing \node{i} when a
thread, among the $(\nbth-1)$ remaining threads, is currently executing
a \castxts on the same node.\cut{ As mentioned before, \castxts indeed
  compels the exclusive ownership of the cacheline, thus, any request
  to access this cacheline during the execution of a \castxts has to
  stall until the end of the \castxts execution.} As a first
approximation, supported by the rareness of the events, we assume that
at most one thread will wait for the access to the node.


Firstly, we obtain the rate of \castxts events 
generated by $(\nbth-1)$ threads through the merge of their poisson
processes.
%
Consider a traversal of \node{i} at a random time; (i) the probability
of being stalled is the ratio of time when \node{i} is
occupied by a \castxts of $(\nbth-1)$ threads, given by:
$\rateCas{i}(\nbth-1)\latcas$; (ii) the stall time that the thread
would experience is distributed uniformly in the interval
$[0,\latcas]$. Then, we obtain: $\expec{\cassta{i}} =
\rateCas{i}(\nbth-1)\latcas(\latcas/2)$.


\compsubsubsection{Invalidation Recovery}
\label{sec:recovery}

Given a thread, a coherence cache miss occurs if \node{i} is modified
by any
other
thread in between two consecutive traversals \cut{(could be \Read or
  \castxts events) }of \node{i}. The events that are concerned are:
(i) the \castxts events from any thread; (ii) the \Read events from the
given thread. When \node{i} is traversed, we look back at these
events, and if among them, the last event was a \castxts from another
thread, a coherence miss occur:
$\proba{\text{Coherence Miss on
    \node{i}}} = \frac{\rateCas{i}(\nbth-1)}{\rateCas{i}\nbth+
  \rateRead{i}}$.
We derive the expected latency of this factor during a traversal at
\node{k} by multiplying this with the latency penalty of a coherence
cache miss: $\expec{\casrec{i}} = \proba{\text{coherence miss on
    \node{i}}} \times \latrec$.


\compsubsubsection{Che's Approximation}
\label{sec:che}

\pr{

Che's Approximation~\cite{che} is a technique to estimate the hit
ratio of an LRU cache of size $C$, where the object (here, node)
accesses follow IRM (Independent Reference Model).  IRM is based on
the assumption that the object references occur in an infinite
sequence from a fixed catalog of $\nbn$ objects. The popularity of
object $i$ (denoted by $\popu{i}$, where $i \in \inti{\nbn}$)
is a constant that does not depend on the reference history and does
not vary over time.

Starting from $t=0$, let the time of reference to object $i$ be
denoted by \obj{i}, then the time for $C$ unique references is given by:
$t_c = \inf\{t>0 : X(t) = C\}$, where $X(t)=\sum_{j=1}^{\nbn}
\mathbf{1}_{0<\obj{j}\leq t}$.  Che's approximation estimates the hit
ratio of object $i$ with $\hitsim{i} \approx 1 - e^{-\popu{i} T}$,
where the so-called characteristic time (denoted by $T$ that
approximates $t_c$) is the unique solution of the following equation:
$ C =\sum_{j=1}^{\nbn} (1 - e^{-\popu{j} T})= \expec{X(T)}$.  The
accuracy of the approximation is rooted at the random variable $X(t)$
that is approximately Gaussian since it is defined as the sum of many
independent random variables (\clt). We provide a more detailed
discussion on this approximation in~\cite{long-version} 
based on the analysis in~\cite{versatile}.}
{


Che's Approximation is a technique to estimate the hit ratio of a LRU
cache, where the object (nodes for our case) accesses follow IRM
(Independent Reference Model).  Che's approximation is concerned with
the capacity misses in a cache.  We apply the approximation to the
\sdss to estimate \expec{\hitcache{i}{\ell}} and \expec{\hittlb{i}{\ell}}. In this
part, we give a brief discussion on Che's Approximation and in the
following sections (see~\ref{sec:cachemiss}, \ref{sec:pagemiss}), we
have shown how we adapt this scheme for our purposes.

IRM is based on the assumption that the object references occur in an
infinite sequence from a fixed catalog of \nbn objects.  The
probability of referencing object $i$ at any point in the sequence
(denoted by \popu{i}, where $i \in \inti{\nbn}$) 
is a constant that does not depend on the reference history and does
not vary over time.
Under LRU policy with cache of size \casi{\ell} and subject to IRM demand of
\nbn objects, an object reference would lead to a capacity miss if at
least \casi{\ell} unique object references take place after the previous
reference to the same object.  Let a reference to object $i$ (\obj{i})
occurs at time $t_0$, the characteristic time for the object $i$ is
defined by the random variable:

\[ \chatio{\ell}{i} = \inf\{t>0 : \xcaf{i}{t} = \casi{\ell}  \} \mathrm{, where,}\]

\[ \xcaf{i}{t} =\sum_{j=1,j \neq i}^{\nbn} \mathbf{1}_{t_0<\obj{j}\leq t} \]

Briefly, Che's approximation, first combines all \chatio{\ell}{i},
where $i \in \inti{\nbn}$ in a single variable by assuming \popu{i} is
negligible compared to $\sum_{j=1}^{\nbn} \popu{j}$ and then
approximates \chatio{\ell}{i} with a constant \chati{\ell} over
objects. Consider a sequence of references that follows an IRM demand
for \nbn objects, with reference probability \popu{i}, where $i \in
\inti{\nbn}$. The characteristic time \chati{\ell} of a cache with size
\casi{\ell} is the unique solution of the following equation:

\[ \casi{\ell} = \sum_{i=1}^{\nbn} (1 - e^{-\popu{i} \chati{\ell}}) \]


In~\cite{versatile}, they analyse and illustrate the reason behind the
accuracy of the approximations for a quite large spectrum of object
reference distributions.  Their argument relies on the random variable
$\xcas{t} = \sum_{j=1}^{\nbn} \mathbf{1}_{t_0<\obj{j}\leq t}$, that provides the
number of unique object references that have occured in the interval
$[0,t]$. As the crucial property, \xcas{t} is defined as the sum of
independent random variables. Based on the central limit theorem, they
show that a Gaussian approximation for this sum is quite reasonable,
for all $t$.

Without loss of generality, let an object $i$ is referenced
consecutively at time $0$ and $t$.  We know that the second reference
would be cache miss, in a cache of size \casi{\ell}, if $X(t) >
\casi{\ell}$, where by assumption \xcas{t} is a Gaussion random
variable. The cache hit ratio of cacheline is given by:

\begin{equation} 
\label{eq:3}
\hito{\ell}{i}  = 1 - \int_{0}^{+\infty}  \proba{\xcas{t} > \casi{\ell}} \popu{i} e^{-\popu{i} t}dt
\end{equation}

Che's approximation, basically, approximates the cumulative distribution
function of \xcas{t} with a step function that cuts this S-shaped
cumulative distribution function at the $ \expec{\xcas{t}}=\sum_{i=1}^{\nbn}  
(1 - e^{-\popu{i}t}) $, denoted by $m(t)$. Thus, it approximates 
\hito{\ell}{i} in Equation~\ref{eq:3} with:
\begin{align*}
 \hito{\ell}{i}  & \approx 1 - \int_{0}^{+\infty} \mathbf{1}_{m(t)>\casi{\ell}} \popu{i} e^{-\popu{i} t}dt\\
		& = 1 - \int_{0}^{+\infty} \mathbf{1}_{t>\chati{\ell}} \popu{i} e^{-\popu{i} t}dt
\end{align*}

In this study, we have exploited Che's approximation to estimate the data 
and TLB cache hit ratios with a slight modification by keeping our 
arguments along the same lines with the ones presented above.




}

\compsubsubsection{Cache Misses}
\label{sec:cachemiss}

We consider a data cache at level $\ell$ of size \casi{\ell} and
compute the hit latency due to \Read events on this cache. We assume
that \node{i} is either present in the \sds or not, during the
characteristic time of the cache.
\Read events at \node{i} are indeed much more frequent than the
removal or insertion of \node{i}. This implies that if the
characteristic time is long enough to accommodate the intervals where
$\node{i} \in \DS$ and $\node{i} \not\in \DS$, then the cache miss
ratio of \node{i} should be quite low, which would be underestimated
due to our assumption. We can employ the \Read rates as popularities,
\ie $\popu{i} = \rateRead{i}$, and modify Che's approximation to
discriminate whether, at a random time, \node{i} is inside the \ds
or not.

We integrate the masking variable \pres{i} into Che's approximation.
We have: $X^{\fcache}(t)=\sum_{i=1}^{\nbn}
\pres{i}\step{0 < \obj{i} \leq t}$, where \obj{i} denotes the reference time
of \node{i}.
We can still assume $X^{\fcache}(t)$ is gaussian, as a sum of many
independent random variables. We estimate the characteristic time as
follows with the linearity of expectation and the independence of the
random variables:
\pr{$\expec{X^{\fcache}(t)} = \sum_{i=1}^{\nbn} \expec{\pres{i}{ \step{0<\obj{i}\leq t}}}
                = \sum_{i=1}^{\nbn} \expec{\pres{i}}\expec{\step{0<\obj{i}\leq t}}
		= \sum_{i=1}^{\nbn} \prespro{i}(1 - e^{-\rateRead{i} t})$.}
   {
     \[ \expec{X^{\fcache}(t)} = \sum_{i=1}^{\nbn} \expec{\pres{i}{ \step{0<\obj{i}\leq t}}}
                = \sum_{i=1}^{\nbn} \expec{\pres{i}}\expec{\step{0<\obj{i}\leq t}}
		= \sum_{i=1}^{\nbn} \prespro{i}(1 - e^{-\rateRead{i} t}).\]
     }
Lastly, we solve the equation for the characteristic time \chati{\ell}
of level $\ell$ cache: $\sum_{i=1}^{\nbn} \prespro{i}(1 - e^{-\rateRead{i}
  \chati{\ell}}) = \casi{\ell}$ thanks to a fixed-point approach.
After computing \chati{\ell}, we estimate the
cache hit ratio (on level $\ell$) of \node{i}:
$1 - e^{-\rateRead{i} \chati{\ell}}$.


\compsubsubsection{Page Misses}
\label{sec:pagemiss}

In this paragraph, we aim at computing the page hit ratio of \node{i}
for the TLB cache at level $\ell$ of size \casit{\ell}. The total
number \nbpg of pages that are used by the \sds can be regulated by a
parameter of the memory managements scheme (frequency of recycling
attempts for the deleted nodes), as the total number of nodes is a
function of \range. 
Different from the
cachelines (corresponding to the nodes), we can safely assume that a
page accommodates at least a single node that is present in the \shds
at any time.


We cannot apply straightforwardly Che's approximation since the page
reference probabilities are unknown. 
However, we are given the cacheline reference
probabilities $\popu{i} = \rateRead{i}$
for $i \in \inti{\nbn}$ and we assume that \nbn cachelines are mapped
uniformly to \nbpg pages, $\inti{\nbn} \to \inti{\nbpg}$, $\nbn >
\nbpg$. Under these assumptions, we know that the resulting page
references would follow IRM because aggregated Poisson processes form
again a poisson process.


\cut{We need to project the skewness in cacheline reference
  probabilities to our page references to obtain accurate
  estimations. }We follow the same line of reasoning as in the cache
miss estimation. First, we consider a set of Bernoulli random
variables $(\prs{i}{j})$, leading to a success if \node{i} is mapped
into page $j$, with probability $\prespro{i}/\nbpg$ (hence \prs{i}{j}
does not depend on $j$).  Under IRM, we can then express the page
references as point processes with rate $r_j = \sum_{i=1}^{\nbn}
\prs{i}{j} \popu{i}$, for all $j \in \inti{\nbpg}$.

%
%
		
Similar to the previous section, we denote the time of a reference to
page $j$ with \obj{j} and we define the random variable
$X^{page}(t)=\sum_{j=1}^{\nbpg} \mathbf{1}_{0<\obj{j}\leq t}$ and compute its expectation:
\begin{align*}
\expec{X^{page}(t)} & = \sum_{j=1}^{\nbpg} \expec{\mathbf{1}_{0<\obj{j}\leq t}}
		 = \sum_{j=1}^{\nbpg} \expec{1 - e^{-r_j t}}
		 = \sum_{j=1}^{\nbpg} \expec{1 - e^{-\sum_{i=1}^{\nbn} \prs{i}{j} \rateRead{i} t}}\\
		& = \sum_{j=1}^{\nbpg} \left( 1 - \prod_{i=1}^{\nbn} \expec{e^{-\prs{i}{j} \rateRead{i} t}} \right)
		= \sum_{j=1}^{\nbpg} \left(1 - \prod_{i=1}^{\nbn}  \left(\frac{\nbpg-\prespro{i}}{\nbpg} + \frac{\prespro{i} e^{-\rateRead{i} t}}{\nbpg}\right) \right)\\
\expec{X^{page}(t)} & = \nbpg \left(1 - \prod_{i=1}^{\nbn}  \left(\frac{\nbpg-\prespro{i}}{\nbpg} + \frac{\prespro{i}e^{-\rateRead{i} t}}{\nbpg}\right) \right),
\end{align*}

Assuming $X^{page}(t)$ is Gaussian as it is sum of many independent 
random variables,
we solve
the following equation for the constant \chatit{\ell} (characteristic
time of a TLB cache of size $C$):
$\expec{X^{page}(\chatit{\ell})}=\casit{\ell}$.

Lastly, we obtain the TLB hit rate for \node{i} by relying on the
average \Read rate of the page that \node{i} belongs to; we should add 
to the contributions of \node{i}, the references to
of the nodes that belong to the same page as \node{i}. 
Then follows the TLB hit ratio:
$1 - e^{-z_{i}\chatit{\ell}}$, where
\pr{$z_i = \rateRead{i} +
  \expec{\sum_{j=1, j \neq i}^{\nbn} \prs{j}{k} \rateRead{j}}= \rateRead{i} +
  \sum_{j=1, j \neq i}^{\nbn} \prespro{j}\rateRead{j}/\nbpg$.}
   {
     \[z_i = \rateRead{i} +
  \expec{\sum_{j=1, j \neq i}^{\nbn} \prs{j}{k} \rateRead{j}}= \rateRead{i} +
  \sum_{j=1, j \neq i}^{\nbn} \prespro{j}\rateRead{j}/\nbpg.\]
     }

\compsubsubsection{Interactions}

To be complete, we mention the interaction between impacting 
factors and the possibility of latency overlaps in the pipeline. 
Firstly, the traversal latency of different nodes can not be overlapped 
due to the semantic dependency for the linked nodes. For a single node traversal, 
the latency for \Cas execution and stall time can not be overlapped with any other 
factor. We consider inclusive data and TLB caches. It is not possible to have a 
cache hit on level $l$, if the cache on level $l-1$ is hit, and we do not 
consider any cost for the data cache hit if invalidation recovery 
(coherence) cost is induced 
(\ie $\expec{\hitcache{i}{\ell}} = (1-\proba{coherence\text{ }miss}) 
(\proba{hit\text{ }cache_l}-\proba{hit\text{ }cache_{l-1}}) 
\latdat{\ell}$).\pp{\vsme}

\subsection{Latency vs. Throughput}
\label{sec:little}

In the previous sections, we have shown how to compute the expected
traversal latency for a given node. There remains to combine these
traversal latencies in order to obtain the throughput of the \sds.  Given
$\node{i}\in\DS$, the average arrival rate of threads to \node{i} is
$\rateTra{i}=\rateRead{i}+\rateCas{i}$. Thus the average arrival rate
of threads to \node{i} is: $\prespro{i}\rateTra{i}$.  It can then be
passed to Little's Law~\cite{littles-law}, which states that the
expected number of threads (denoted by $t_i$) traversing \node{i} obeys to
$t_i = \prespro{i}\rateTra{i}\expec{Traverse_{i}}$.
The equation holds for any node in the \sds, and for the application
call occurring in between \sds operations. Its expected latency is a
parameter ($\expec{\trav{0}}=\latapp$) and its average arrival rate
 is equal to the throughput ($\rateTra{0}=\thr$). Then, we
have: $\sum_{i=0}^{\nbn} t_i = \sum_{i=0}^{\nbn}
(\prespro{i}\rateTra{i}\expec{\trav{i}})$, where \rateTra{i} and
\expec{Traverse_{i}} are linear functions of $\thr$. We also know
$\sum_{i=0}^{\nbn} t_i = \nbth$ as the threads should be executing
some component of the program.  We define constants with $a_i, b_i,
c_i$ for $i \in \inti[0]{\nbn}$. And, we represent $\rateTra{i} =
a_i \thr$ and $\expec{Traverse_{i}} = b_i \thr + c_i$ and we obtain
the following second order equation: $\sum_{i=0}^{\nbn} (\prespro{i}
a_i b_i)\thr^2 + \sum_{i=0}^{\nbn} (\prespro{i} a_i c_i) \thr - \nbth
= 0$.
This second order equation
has a unique positive solution that provides the expected throughput,
\thr.\pp{\vsme}

\section{Instantiating the Throughput Model\pp{\vssm}}
\label{sec:dss}

In this section, we show how to initialize our model with widely known
\lf \sdss, that have different operation time complexities.
In order to obtain a throughput estimate for a \shds, we need to
compute the rates \rateRead{\star} and \rateCas{\star}, and
$\proba{\op{o}{k} \leadsto \opmem{\node{i}} | \node{i}\in\DS}$, \ie
the probability that, at a random time, an operation of type $o$ on
key $k$ leads to a memory instruction of type $\formmem$ on node
\node{i}, knowing that \node{i} is in the \ds. For the ease of
notation, nodes will sometimes be doubly or triply indexed, and when
the context is clear, we will omit $|\node{i}\in\DS$ in the
probabilities.

We first estimate the
throughput of \dslls and \dshts, on which we can directly apply our
method, then we move on more involved \sds, namely \dssls and
\dsbts, that need a particular attention.\pp{\vsme}

\subsection{\dslltit}
\label{sec:ll}

We start with the \lf \dsll implementation of Harris~\cite{harris}.
All operations in the \dsll start with the search phase in which the
\dsll is traversed until a key. At this point all operations terminate
except the successful update operations that proceed by modifying a
subset of nodes in the structure with \castxts instructions.  The
\shds contains only valued node and two sentinel nodes \node{0} and
\node{\range+1}, so that $\nbn = \range+2$ and for all $i \in
\inti{\range}$, \node{i} holds key $i$, \ie $\keyv{i} = i$.

First, we need to compute the probabilities of triggering a \Read
event and \castxts event on a node, given that the node is in the
\sds, for all operations of type $t \in \{\ins, \del, \src\}$ 
targeted to key $k$.

At a random time, \node{k}, for $k \in \inti{\range}$, is in the \dsll
iff the last update operation on key $k$ is an insert:
$\prespro{k}=\plin{k}$, by definition of \plin{k}.
Moreover, when \node{k} is in the \shds (condition that we omit in the
notation), \op{t}{k'} reads \node{k}, either if \node{k} is before
\node{k'}, or if it is just after \node{k'}. Formally,
$\proba{\op{o}{k'} \leadsto \opread{\node{k}}} = 1$ if $k \leq k'$ and
$\proba{\op{o}{k'} \leadsto \opread{\node{k}}} = \prod_{i=k'}^{k-1}
(1-\prespro{i})$ if $k>k'$.

\cut{for $k \in \inti{\range}$ and
  $\prespro{0}=\prespro{\range+1}=1$.}

%
%
%

\castxts events can only be triggered by successful \ins and \del
operations. A successful \ins operation, targeted to \node{k'}, is
realized with a \castxts that is executed on \node{k}, where $k =
\sup\{\ell < k': \node{\ell} \in \DS\}$. The probability of success,
which conditions the \castxts's, follows from the presence
probabilities:
\cut{
In other words, it happens at
\node{i} if there are not any nodes in between \node{i} and \node{k}.
A successful \del, targeted to key $k$, is realized with two \castxts
instructions: (i)\cut{executed} at \node{k} to mark it;
(ii)\cut{executed} at \node{i}, where $i = \sup\{\ell < j: \node{\ell}
\in \DS\}$. Furthermore, we obtain the success probabilities of \ins
and \del operations, targeted to key $k$, with $(1-\prespro{k})$ and
\prespro{k}, respectively.  Hence,
}
\setlength\arraycolsep{1.5pt}

\pr{
\begin{minipage}{.4\textwidth}
\small 
\begin{multline}
\proba{\op{\formins}{k'} \leadsto \opcas{\node{k}}} \\
= \left\{
\begin{array}{ll}
0, & \text{if } k \geq k'\\
\displaystyle\prod_{i=k+1}^{k'} (1-\prespro{i}),
& \text{if } k < k'
\end{array} \right.\quad ;
 \end{multline}
\end{minipage}\hfill
\begin{minipage}{.4\linewidth}
\small
\begin{multline}
\proba{\op{\formdel}{k'} \leadsto \opcas{\node{k}}} \\ 
= \left\{
\begin{array}{ll}
1, & \text{if } k = k'\\
0, & \text{if } k > k'\\
\prespro{k'}\displaystyle\prod_{i=k+1}^{k'-1} (1-\prespro{i}),
& \text{if } k < k'
\end{array} \right.
 \end{multline}
\end{minipage}%
}
{
\[ \proba{\op{\formins}{k'} \leadsto \opcas{\node{k}}} 
= \left\{
\begin{array}{ll}
0, & \text{if } k \geq k'\\
\displaystyle\prod_{i=k+1}^{k'} (1-\prespro{i}),
& \text{if } k < k'
\end{array} \right.\quad
\]
\[
\proba{\op{\formdel}{k'} \leadsto \opcas{\node{k}}} \\ 
= \left\{
\begin{array}{ll}
1, & \text{if } k = k'\\
0, & \text{if } k > k'\\
\prespro{k'}\displaystyle\prod_{i=k+1}^{k'-1} (1-\prespro{i}),
& \text{if } k < k'
\end{array} \right.
\]
}

\subsection{\vssm{}\dshttit}
\label{sec:ht}

We analyse here a chaining based hash table where elements are hashed
to $B$ buckets implemented with the \lf \dsll of
Harris~\cite{harris}. The \shds is parametrized with a load factor
\loadf which determines $B$ through $B = \range/\loadf$. The hash
function $h:k \mapsto \ceil{k/\loadf}$ maps the keys sequentially to
the buckets, so that, after including the sentinel nodes (2 per
bucket), we can doubly index the nodes: \node{b,k} is the node in
bucket $b$ with key $k$, where $b \in \inti{B}$ and $k \in
\inti{\loadf}$ (the last bucket may contain less elements).\pp{\vsbi}


\pr{

\begin{center}\small
\begin{align*}
\proba{\op{o}{b',k'} \leadsto \opread{\node{b,k}}} =
\left\{\begin{array}{ll}
0, & \text{if } b' \neq b \\
1, & \text{if } b'=b \text{ and } k' \geq k\\
\displaystyle\prod_{j=k'}^{k-1} (1-\prespro{b,j}), & \text{if } b'=b \text{ and } k'<k
\end{array}\right.
\end{align*}\end{center}

\begin{minipage}{.4\textwidth}
  \small
\begin{multline*}
\proba{\op{\formins}{b',k'} \leadsto \opcas{\node{b,k}}} =\\
\left\{\begin{array}{ll}
0, & \text{if } b' \neq b \text{ or } k' \leq k\\
\displaystyle\prod_{j=k+1}^{k'} (1-\prespro{b,j}), & \text{if } b'=b \text{ and }k'>k
\end{array}\right.
\end{multline*}
\end{minipage}
\begin{minipage}{.4\textwidth}
\small
\begin{multline*}
\proba{\op{\formdel}{b',k'} \leadsto \opcas{\node{b,k}}} =\\
\left\{\begin{array}{ll}
0, & \text{if } b' \neq b \text{ or } k' < k\\
1, & \text{if } b' = b \text{ and } k'=k \\
\prespro{b,k'}\displaystyle\prod_{j=k+1}^{k'-1} (1-\prespro{b,j}), & \text{if }
b'=b \text{ and } k' > k 
\end{array}\right.
\end{multline*}
\end{minipage}%
}
{
  \[
  \proba{\op{o}{b',k'} \leadsto \opread{\node{b,k}}} =
\left\{\begin{array}{ll}
0, & \text{if } b' \neq b \\
1, & \text{if } b'=b \text{ and } k' \geq k\\
\displaystyle\prod_{j=k'}^{k-1} (1-\prespro{b,j}), & \text{if } b'=b \text{ and } k'<k
\end{array}\right.
\]
\[
\proba{\op{\formins}{b',k'} \leadsto \opcas{\node{b,k}}} =
\left\{\begin{array}{ll}
0, & \text{if } b' \neq b \text{ or } k' \leq k\\
\displaystyle\prod_{j=k+1}^{k'} (1-\prespro{b,j}), & \text{if } b'=b \text{ and }k'>k
\end{array}\right.
\]
\[
\proba{\op{\formdel}{b',k'} \leadsto \opcas{\node{b,k}}} =
\left\{\begin{array}{ll}
0, & \text{if } b' \neq b \text{ or } k' < k\\
1, & \text{if } b' = b \text{ and } k'=k \\
\prespro{b,k'}\displaystyle\prod_{j=k+1}^{k'-1} (1-\prespro{b,j}), & \text{if }
b'=b \text{ and } k' > k 
\end{array}\right.
\]
}


\medskip

In the previous two \dss, we do observe differences in the traversal
rate from node to node, but the node associated with a given key does
not show significant variation in its traversal rate during the course
of the execution: inside the \shds, the number of nodes preceding (and
following) this node is indeed rather stable. In the next two \dss,
node traversal rates can change dramatically according to node
characteristics, that may include its position in the \shds. In a
\dssl, a node \node{i} containing key \keyv{i} with maximum height will be
traversed by any operation targeting a node with a higher
key. However, \node{i} can later be deleted and
inserted back with the minimum height; the operations that traverse it
will then be extremely rare. The same reasoning holds when comparing
an internal node with key \keyv{i} of a \dsbt located at the root or
close to the leaves.

As explained before, an accurate cache miss analysis cannot be
satisfied with average access rates. Therefore, the information on the
possible significant variations of rates should not be diluted into a
single access rate of the node. To avoid that, we pass the information
through virtual nodes: a node of the \shds is divided into a set of
virtual nodes, each of them holding a different flavor of the initial
node (height of the node in the \dssl or subtree size in the
\dsbt). The virtual nodes go through the whole analysis instead of the
initial nodes, before we extract the average behavior of the system
hence throughput.\pp{\vsme}

\subsection{\dssltit\pp{\vssm}}
\label{sec:sl}

There exist various \lf \dssl implementations and we study here the
\lf \dssl~\cite{sundell}. 
\dsSls offer layers of \dslls.  Each layer is a sparser version of the
layer below where the bottom layer is a \dsll that includes all the
elements that are present in the \sds. An element that is present in
the layer at height $h$ appears in layer at height $h+1$ with a fixed
appearance probability ($1/2$ for our case) up to some maximum layer
\hmax that is a parameter of the \dssl.

\dsSl implementations are often realized by distinguishing two type of
nodes: (i) valued nodes reside at the bottom layer and they hold the
key-value pair in addition to the two pointers, one to the next node
at the bottom layer and one to the corresponding routing node (could
be \nulltxt); (ii) routing nodes are used to route the threads towards
the search key. Being coupled with a valued node, a routing node does
not replicate the key-value pair. Instead, only a set of pointers,
corresponding to the valued node containing the next key in different
layers, are packed together in a single routing node (that fits in a
cacheline with high probability).
Every \Read event in a routing node is preceded by a \Read in the
corresponding valued node.

We denote by \nodero{k}{h} the routing node containing key $k$, whose
set of pointers is of height $h$, where $h \in \inti{\hmax}$. A valued
node containing the key $k$ is denoted by \nodeda{k}{h} when connected
to \nodero{k}{h} ($h=0$ if there is no routing node). Furthermore,
there are four sentinel nodes \nodeda{0}{\hmax}, \nodero{0}{\hmax},
\nodeda{\range+1}{\hmax}, \nodero{\range+1}{\hmax}.
The presence probabilities result from the coin flips (bounded by
\hmax): for $z \in \{\formdat, \formrou\}$,
$\presprord{z}{k}{h} = 2^{-(h+1)} \plin{k}$ if $h<\hmax$,
$\presprord{z}{k}{h} = \plin{k} - \sum_{\ell=0}^{\hmax-1} \presprord{z}{k}{\ell}$ otherwise.

\begin{figure}%
    \begin{center}{\scalebox{.55}{\input{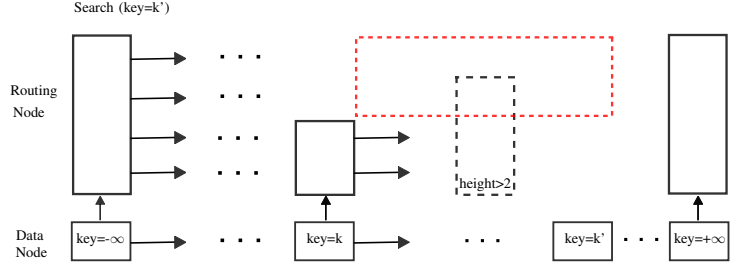}}}\end{center}
    \caption{\dssltit Events: Read Event Probability}%
    \label{fig:sklread}%
\end{figure}

By decomposing into three cases, we compute the probability that an
operation \op{o}{k'} of type $o \in \{\formins, \formdel, \formsrc\}$,
targeted to $k'$, causes a \Read triggering event at
\noderd{z}{k}{h} when $\noderd{z}{k}{h} \in \DS$.
Let assume first that $k'>k$. The operation triggers a \Read event at
node \noderd{z}{k}{h} if for all $(x, y)$ such that $y>h$ and $k < x
\leq k'$, \nodet{z}{x}{y} is not present in the \dssl (\ie in
Figure~\ref{fig:sklread}, no node in the \dssl overlaps with the
red frame).
Let assume now $k'<k$. The occurrence of a \Read event requires that:
for all $(x,y)$ such that $y\geq h$ and $k' \leq x < k$,
\noderd{z}{x}{y}, is not present in the \shds.
Lastly, a \Read event is certainly triggered if $k'=k$. The final
formula is given by:\pp{\vsbi{}\vsbi}
%

{\small \[
\proba{\op{o}{k'} \leadsto \opread{\noderd{z}{k}{h}}} \
= \begin{cases}
\prod_{x=k+1}^{k'} \left(1 - \left(\sum_{y=h+1}^{\hmax} \presprord{z}{x}{y}\right)\right), & \text{if } k \leq k'\\
\prod_{x=k'}^{k-1} \left(1 - \left(\sum_{y=h}^{\hmax} \presprord{z}{x}{y}\right)\right), & \text{if } k > k'
\end{cases}
\]\pp\vsmb}

\pr{
To be complete, we describe in~\cite{long-version} how to
compute the probability for \castxts events, following a
similar approach.\vsme
}
{


\begin{figure}%
\begin{center}{\scalebox{.55} {\input{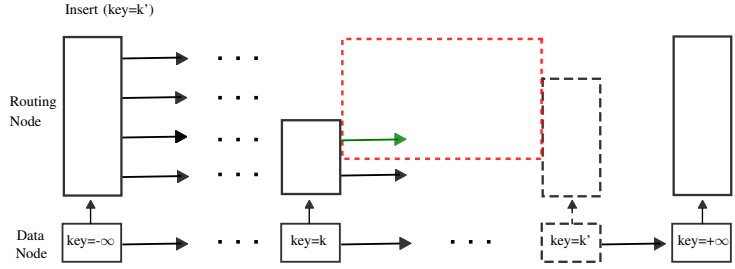}}}\end{center}
    \caption{Skiplist Events: \castxts Event Probability}%
    \label{fig:sklcas}%
\end{figure}

Next, we apply a similar approach for \castxts events.  In
Figure~\ref{fig:sklcas}, we illustrate an example. A \castxts event
occurs at the green pointer, as a result of the removal (or insertion)
of \keyv{k} if there is no node in the red frame.  For all node and
operation couples, $\proba{\op{o}{k'} \leadsto \opcas{\noderd{z}{k}{h}}}$
is simply obtained in those lines.

The insertion of an element with \keyv{k'} introduces $\nodet{z}{k'}{h}$
with probability $2^{-(h+1)}$ if $h \in \inti{\hmax-1}$, and $1 -
\sum_{i=0}^{\hmax-1} 2^{-(h+1)}$ when the maximum height. The data node is linked
to the list at the bottom layer with a \castxts that is executed on
the previous data node.  If a routing node is introduced, it is linked
to lists at $h$ different layers, thus leads to $h$ \castxts
instructions that are applied on the other nodes.

The deletion of an element is composed of two phases. The first phase
is to mark the data node, $\nodeda{k'}{h}$ and the pointers in the
routing node with height $k'$, if it exists. If the height of the
routing node is more than one, it is possible that multiple \castxts
intructions are executed on the same routing node. But, we only
consider the first one. The latency and also the effect of remaining
ones would be negligible, as they are applied on the same cacheline
one after each other. This repetitive behavior guarentees that the
cacheline has already been exclusively owned before the next \castxts
instructions run. To recall, this is consistent with our assumption
that an event can occur at most once per operation on a node. The
second phase of deletion operation follows the same path with the
insertion operation. Simply, a \castxts, on the previous node, is
executed for each layer that the data and routing nodes span.

We have denoted the success probability of an \ins operation with 
$\plin{k'}=\frac{\proba{op=\op{\ins}{k'}}}{\proba{op=\op{\ins}{k'}}+\proba{op=\op{\del}{k'}}}$. 
Also, the factor $2^{-(h+1)}$ provides the probability of the insertion of a
routing node with height $h$, coupled with its data node. Based on the
non-existence of any node that overlaps with the area that is enclosed
with the red frame in Figure~\ref{fig:sklcas}, we obtain:

\[
\proba{\op{\formins}{k'} \leadsto \opcas{\nodet{z}{k}{h}}} =\\
\begin{cases}
(1-\plin{k'})(\sum_{h=0}^{\hmax} 2^{-(h+1)} (\prod_{x=k+1}^{k'-1} (1 - (\sum_{y=h}^{\hmax} \presprord{z}{x}{y} ))))  , & \text{if } k < k'\\
0, & \text{if } k \geq k' \\
\end{cases}
\]

\[
\proba{\op{\formdel}{k'} \leadsto \opcas{\nodet{z}{k}{h}}} =\\
\begin{cases}
1, & \text{if } k = k'\\
\plin{k'} (\sum_{h=0}^{\hmax} 2^{-(h+1)} (\prod_{x=i+1}^{k'-1} (1 - (\sum_{y=h}^{\hmax} \presprord{z}{x}{y}))))  , & \text{if } k < k'\\
0, & \text{if } k > k' \\
\end{cases}
\]

%
%
%
%
%
%
%
%
%
%
%
%
%
%

}

\subsection{\dsbttit\pp\vssm}
\label{sec:bt}

\newcommand{\floor}[1]{\lfloor #1 \rfloor}
\newcommand{\pri}[1]{prio_#1}

We show here how to estimate the throughput of external \dsbts. They
are composed of two types of nodes: internal nodes route the search
towards the leaves (routing nodes) and store just a key, while leaves,
referred as external nodes contain the key-value pair (valued
node). We use the external \dsbt of Natarajan~\cite{natarajan} to
initialize our model. The search traversal starts and continues with a
set of internal nodes and ends with an external node. We denote by
\nodeint{k} (resp. \nodeext{k}) the internal (resp. external) node
containing key $k$, where $k \in \inti{\range}$. The tree contains
two sentinel internal nodes that reside at the top of the tree (hence
are traversed by all operation): \nodeint{-1} and \nodeint{0}.

Our first aim is to find the paths followed by any operation through
the \dsbt, in order to obtain the access triggering rates, thanks to
Equation~\ref{eq:lambdas}. \dsBts are more complex than the previous
\shdss since the order of the operations impact the positioning of the
nodes. The random permutation model proposes a framework for
randomized constructions in which we can develop our model. Each key
is associated with a priority, which determines its insertion order:
the key with the highest priority is inserted first. The performance
characteristics of the randomized \dsbts are studied
in~\cite{tree}. In the same vein, we compute the traversal probability
of the internal node with key $k$ in an operation that targets
key $k'$.

\begin{lemma}
\label{th:pint}
Given an external \dsbt, the probability of traversing \nodeint{k} in
an operation that targets key \keyv{k'} is given by: (i) $1/f(k,k')$
if $k' \geq k$; (ii) $1/(f(k',k)-1)$ if $k' < k$, where $f(x,y)$
provides the number internal nodes whose keys are in the interval
$[x,y]$.
\end{lemma}
\begin{proof}
\nodeint{k} would be traversed if it is on the search path to the
external node with key $k'$. Given $k' \geq k$, this happens iff
\nodeint{k} has the highest priority among the internal nodes in the
interval $[k,k']$. This interval contains $f(k,k')$ internal nodes, thus,
the probability of \nodeint{k} to possess the highest priority is
$1/f(k,k')$. Similarly, if $k' < k$, then \nodeint{k} is traversed iff
it has the highest priority in the interval $(k',k]$. Hence, the lemma.
\end{proof}

Even if in the \dsbt, nodes are inserted and deleted an infinite number
of times, Lemma~\ref{th:pint} can still be of use.
The number of internal nodes in the interval $[k,k']$ (or $(k',k]$ if
  $k' < k$)
is indeed a random variable which is the sum of independent Bernoulli
random variables that models the presence of the nodes. As a sum of
many independent Bernoulli variables, the outcome is expected to have
low variations because of its asymptotic normality. Therefore, we
replace this random variable with its expected value and stick to this
approximation in the rest of this section. The number of internal
nodes in any interval come out from the presence probabilities:
$\presproie{k}{z} = \plin{k}$, where $z \in \{\formint, \formext\}$.

In an operation is targeted to key $k'$, a single external node is
traversed (if any): \nodeext{k'}, if present, else the external node
with the biggest key smaller than $k'$, if it exists, else the
external node with the smallest key.  Then, we have:\pp\vsbi

\pr{\begin{minipage}{.45\linewidth}
\begin{multline*}
\proba{\op{o}{k'} \leadsto \opread{\nodeint{k}}}\\
= \left\{
\begin{array}{lcl}
1/(1 + \sum_{i=k'+1}^{k-1} \presproint{i}), & \text{if }  k > k'\\
1/(1 + \sum_{i=k+1}^{k'} \presproint{i}), & \text{if }  k \leq k'
\end{array} \right.,
\end{multline*}
\end{minipage}%
\begin{minipage}{.5\linewidth}
\begin{multline*}
\proba{\op{o}{k'} \leadsto \opread{\nodeext{k}}}\\
= \left\{
\begin{array}{lcl}
1, & \text{if } k=k'\\
\prod_{i=k+1}^{k'} (1-\presproext{i}), & \text{if } k<k'\\
\prod_{i=1}^{k-1} (1-\presproext{i}) , & \text{if } k>k'
\end{array} \right.
\end{multline*}
\end{minipage}%
}{
\[ \proba{\op{o}{k'} \leadsto \opread{\nodeint{k}}}
= \left\{
\begin{array}{lcl}
1/(1 + \sum_{i=k'+1}^{k-1} \presproint{i}), & \text{if }  k > k'\\
1/(1 + \sum_{i=k+1}^{k'} \presproint{i}), & \text{if }  k \leq k'
\end{array} \right.,
\]
\[ \proba{\op{o}{k'} \leadsto \opread{\nodeext{k}}}
= \left\{
\begin{array}{lcl}
1, & \text{if } k=k'\\
\prod_{i=k+1}^{k'} (1-\presproext{i}), & \text{if } k<k'\\
\prod_{i=1}^{k-1} (1-\presproext{i}) , & \text{if } k>k'
\end{array} \right.
\]
}

\vspace{.2cm}

These probabilities finally lead to the computation of the \Read
(resp. \castxts) rates \rateReadIE{z}{k} (resp. \rateCasEIA{z}{k}) of
$\noded{k}{z}$, where $z \in \{ \fint, \fext\}$, that will be used in the
last following step. 

We focus now on the \Read rate of the internal nodes. We have found
the average behavior of each node in the previous step; however, the
node can follow different behaviors during the execution since the
\Read rate of \nodeint{k} depends on the size of the subtree whose
root is \nodeint{k}, which is expected to vary with the update
operations on the tree. We dig more into this and reflect these
variations by decomposing \nodeint{k} into \nbvirt{k} virtual nodes,
\nodeintv{k}{h}, where $h \in \inti{\nbvirt{k}}$.  We define the \Read
rate \rateReadIntv{k}{h} of these virtual nodes as a weighted sum of the
initial node rate thanks the two equations
$\presproint{k}=\sum_{h=1}^{\nbvirt{k}} \presprointv{k}{h}$ and
$\presproint{k} \rateReadInt{k} = \sum_{h=1}^{\nbvirt{k}}
\presprointv{k}{h} \rateReadIntv{k}{h}$.

We connect the virtual nodes to the initial nodes in two ways. On the
one hand, one can remark that the \Read rate is proportional to the
subtree size: $\rateReadIntv{k}{h} \propto h \rateReadInt{k}$. On the
other hand, based on the probability mass function of the random
variable \subt{k} representing the size of the subtree rooted at
\nodeint{k}, we can evaluate the weight of the virtual nodes:
$\presprointv{k}{h}=\presproint{k} \proba{\subt{k}=h}$.

\pr{

We follow a similar approach for the other access rates, and use them
in the rest of the model to end up with the \ds throughput. More
details are to be found in~\cite{long-version},
especially on how to obtain this mass function and how to deal with
the \castxts events.\pp\vsbi
}
{



We have computed \rateReadInt{k}. These values reflect the average behaviour 
along the whole execution.
However, the average behavior is not enough to computethe traversal latency 
accurately for the internal nodes. In the execution,
there are different time intervals where \rateReadInt{k} show significant variation 
depending on the part of the tree that it is located.
For instance, it is quite improbable to observe a cache miss at \nodeint{k} when it
is positioned at the root of the tree.
One would observe a very high rate of traversals
with low latency in this case, which decreases the expected traversal 
latency of \nodeint{k} significantly. An accurate estimation for the cache 
misses requires the consideration of this particularity of the binary tree.
To approximate the impact of this variation, we split \nodeint{k} into a
number (let $\nbvirt{k}$ denotes this number for \nodeint{k}) of independent virtual nodes (in 
the lines of independent reference model), each representing the behavior 
of \nodeint{k} with a different \Read rate.
The virtual node, with \Read rate \rateReadIntv{k}{h}, is 
denoted by \nodet{k}{h}{int}. We will obtain
the \Read rates \rateReadIntv{k}{h} and presence probabilities
\presprointv{k}{h} for these virtual
nodes by requiring that the average behaviors are still valid:
$\presproint{k}=\sum_{h=1}^{\nbvirt{k}} \presprointv{k}{h}$ 
and $\presproint{k}\rateReadInt{k} =\sum_{h=1}^{\nbvirt{k}} \presprointv{k}{h} \rateReadIntv{k}{h}$.

\newcommand{\ind}[1]{\ema{\sigma_#1}}
\begin{theorem}
\label{th:2}
For an external binary tree with $N$ internal nodes, generated with the random permutation of insertions, the probability mass function of the size of the subtree (the random 
variable concerns only the number of the internal nodes and denoted by \subt{k}) that is 
rooted at $\nodeint{k}$ is given by: $\proba{\subt{k}= N}=1/N$ and $\proba{\subt{k} = s}=O(1/s^2)$.
\end{theorem}

\begin{proof}

It is clear that $\proba{\subt{k} = N}=1/N$ since it occurs iff $\nodeint{k}$ has the highest
priority among all internal nodes. For the rest, we consider four different cases. 
Let \ind{k} denotes the index of $\nodeint{k}$ in the permutation of the sequence of $N$ 
internal nodes that are arranged in the ascending order based on their keys.
 
(i) $\ind{k}+s \leq N$ and $\ind{k}-s \geq 1$:  then there exist $s$ distinct pairs 
of $(\nodeint{j},\nodeint{i})$ such that
$\ind{i}-\ind{j}=s+1$ and $\ind{j}<\ind{k}<\ind{i}$. Given a pair of such 
$(\nodeint{j},\nodeint{i})$, $\subt{k}=s$ if the priorities of
$\nodeint{j}$ and $\nodeint{i}$ are higher than the priorities of all \nodeint{x}, 
such that \ind{j} < \ind{x} < \ind{i} and also $\nodeint{k}$ has a higher priority than 
all \nodeint{y \neq k} such that \ind{j} < \ind{y} < \ind{i}.
This ($\nodeint{k}$ is the root of subtree that includes all \nodeint{y}, 
such that \ind{j} < \ind{y} < \ind{i}) can happen with probability,
$\frac{2}{(s+2)(s+1)s}$. There exist $s$ such non-overlapping cases.
We have, $\proba{\subt{k}=s} = \frac{2}{(s+1)(s+2)}$.

(ii) $\ind{k}+s > N$ and $\ind{k}-s \geq 1$:  then there exist a \nodeint{i} such that 
$\ind{i}=N-s$. $\subt{k}=s$ if $\nodeint{i}$ has higher priority than all \nodeint{x}, 
such that $\ind{i} < \ind{x} \leq N$ and $\nodeint{k}$ has higher priority than all 
\nodeint{y}, such that $\ind{i} < \ind{{y \neq k}} \leq N$. This can happen with probability,
$\frac{1}{(s+1)s}$. In addition, there can be at least $0$ and at most $s-1$ distinct pairs
of $(\nodeint{j},\nodeint{i})$ such that $\ind{i}-\ind{j}=s+1$ and $\ind{j}<\ind{k}<\ind{i}$.
We have: $\frac{1}{(s+1)s} \leq \proba{\subt{k}=s} 
\leq \frac{1}{(s+1)s} + \frac{2(s-1)}{(s+1)(s+2)s}$.

(iii) $\ind{k}+s \leq N$ and $\ind{k}-s < 1$: The bound at (ii) applies to this case also.

(iv) $\ind{k}+s > N$ and $\ind{k}-s < 1$: then there exist a $\nodeint{i}$ such that 
$\ind{i}=N-s$ and a $\nodeint{j}$ such that
$\ind{j}=s+1$. In addition, there can be at least $0$ and at most $s-2$ distinct pairs of nodes 
$(\nodeint{j},\nodeint{i})$ such that $\ind{i}-\ind{j}=s+1$ and $\ind{j}<\ind{k}<\ind{i}$. 
Similar to (i) and (ii), we obtain and sum the probabilities lead to $\subt{k}=s$. We have:
$\frac{2}{(s+1)s} \leq \proba{\subt{k}=s} \leq \frac{2}{(s+1)s} + \frac{2(s-2)}{(s+1)(s+2)s}$
\end{proof}

We start with an observation. The \Read rate of $\nodeint{k}$ is
proportional to the size of the subtree that is rooted at
$\nodeint{k}$. Given a binary tree of $N$ internal nodes, the size
of the subtree can vary in the interval $[1,N]$, which means that we
can have $\nbvirt{k}=N$ different \Read rate levels (\rateReadIntv{k}{h})
associated with their presence probabilities 
$\presprointv{k}{h}=\presproint{k} \proba{\subt{k}=h}$.
Relying on Theorem~\ref{th:2}, one can observe that
$\proba{\subt{k}=h}$ do not variate much from $c_1/(h+1)^2$ for the majority
of different values of $h$ and $k$.  Therefore, we approximate
$\proba{\subt{k}=h} \approx c_1/(h+1)^2$, with a single constant $c_1$ for all
$k$ and $h<\nbvirt{k}$.  We know, $\sum_{h=1}^{\nbvirt{k}} \proba{\subt{k}=h}=1$ and
$\proba{\subt{k}=\nbvirt{k}}=1/\nbvirt{k}$.  So, we obtain $c_1 \approx 2$ by solving the
equation $\int_{h=2}^{N} (c_1/h^2) dh = (N-1)/N$.  We set
$ \presprointv{k}{h} =  \presproint{k} (2/(h+1)^2)$ 
and $\presprointv{k}{\nbvirt{k}} = \presproint{k} / \nbvirt{k}$. 
Assuming $ \rateReadIntv{k}{h} =
c_2 h \rateReadInt{k}$ (\Read rates are proportional to the subtree
size), we require $\presproint{k} \rateReadInt{k} = 
\sum_{h=1}^{\nbvirt{k}} \presprointv{k}{h} \rateReadIntv{k}{h}$,
which leads to $\rateReadInt{k} \approx c_2 +
\int_{h=2}^{\nbvirt{k}} (2/h^2) c_2 (h-1) \rateReadInt{k} dh$. We solve and obtain
$c_2 \approx 1/(2\ln \nbvirt{k})$. We set $\rateReadIntv{k}{h} = h \rateReadInt{k}
/(2\ln \nbvirt{k})$, for the virtual internal nodes.


Now, we consider the \castxts events. \del and \ins operation start
with the search phase. \ins operation finalize with a \castxts
executed at the grandparent internal node of the inserted external
key. \del operation contains three \castxts;
(i) one at the grandparent internal node of the deleted external key;
(ii) two that are executed consecutively at the parent node of the
external key. Thus, we consider them as a single \castxts instruction,
since the second of the consecutive ones has a negligible cost because
the cacheline has already been exclusively owned by the thread.

Similar to \Read events, we first find the rate of \castxts events for
$\nodeint{k}$ and split these events to virtual nodes by requiring 
the average behavior is still valid:
$\presproint{k} \rateCasInt{k} = \sum_{h=1}^{\nbvirt{k}} 
\presprointv{k}{h} \rateCasEI{k}{h}$.
To determine the target of \castxts event, we need to determine the
probability of an internal node $\nodeint{k}$ to be the grandparent
or parent of the targetted $\nodeext{k'}$.  We examine four
different cases as illustrated in Figure~\ref{fig:bstcas}.  Given that
we are in the first case, we look for the probability that
$\nodeint{k}, k'<k$, to possess the smallest or second smallest key,
that is bigger than $k'$, among the internal nodes that are present in
the tree. Such internal nodes with the smallest key and the second
smallest key corresponds to the parent and grandparent of
$\nodeext{k'}$, respectively.
For case 1, it is possible that the grandparent node is the node which 
has the $x$th, $x>1$, smallest key that is bigger than $i$, that is 
present in the tree. But this probability decreases exponentially 
as $x$ increases. That is why, we have attributed the \castxts events
that takes place at the granparent node to the node with second smallest 
key that is bigger than $k'$.
For case 2, the parent corresponds to the
smallest key that is bigger than $k'$ and the grandparent
corresponds to the biggest key that is smaller than $k'$, that are 
present in the tree.

\begin{figure}
\small
\center
\scalebox{.7}{\input{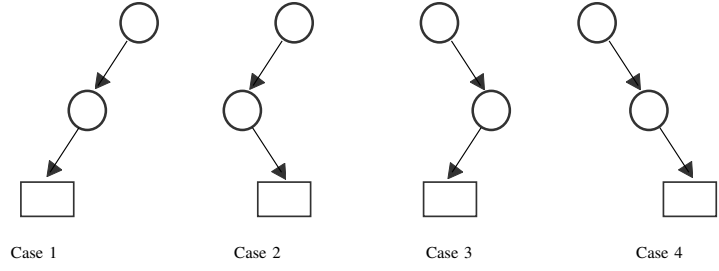}}
\caption{Binary Tree CAS Probability\label{fig:bstcas}}
\end{figure}

Formally, let $P_{k'}^{B} = \{i : i \geq k',  \nodeint{i} \in \DS\}$ and
$P_{k'}^{S} = \{i : i < k',  \nodeint{i} \in \DS\}$.
For the first case, we are interested in the probability that $\nodeint{k}$ is the
grandparent or parent node of $\nodeext{k'}$. These are given by
$\proba{k=\sup\{ P_{k'}^{S}  - \sup\{P_{k'}^{S}\}\}}$ and $\proba{k=\sup\{P_{k'}^{S}\}}$ 
respectively.
For the second case, we are interested
in $\proba{k=\sup\{P_{k'}^{S}\}}$ and $\proba{k=\inf\{P_{k'}^{B}\}}$. The third and fourth cases follows
the same lines as they are the flipped versions of the case one and two.
For all non-sentinel nodes, we have $\presproint{k} = p$. First, we compute the following 
probabilities:

For $k \geq k'$ we have: (these probabilities are zero if $k < k'$)
\[\proba{k=\sup\{P_{k'}^{S} - \sup\{P_{k'}^{S}\}\}} = p (k'-i) (1-p)^{(k'-k-1)}\]
\[\proba{k=\sup\{P_{k'}^{S}\}} = (1-p)^{(k'-k)}\]

And for $k < k'$: (these probabilities are zero if $k \geq k'$)
\[\proba{k=\inf\{P_{k'}^{B}\}} = (1-p)^{(k-k'-1)} \]
\[\proba{k=\inf\{P_{k'}^{B} - \inf\{P_{k'}^{B}\}\}} = p (k-k'-1) (1-p)^{(k-k'-2)}\]

Based on Lemma~\ref{th:pint} (assuming a constant tree size), we obtain the
expected number of internal nodes
that route the search to its left child ($c_{k',l}$) and right child($c_{k',r}$) for an operation
that is targetted to $key=k'$. On this route, we compute the probability of a random node
to be the left (right) child of its parent, with $l_{k'}=c_{k',l}/(c_{k',l}+c_{k',r})$
(and similarly $r_=c_{k',r}/(c_{k',l}+c_{k',r})$). And, we estimate the probability of observing
a case at a random time by using these values (\ie $l_{k'}^2$ for Case 1, $l_{k'}r_{k'}$
for Case 2). And finally, we obtain:

\begin{align*}
\proba{\op{\formdel}{k'} \leadsto \opcas{\nodeint{k}}}=&
\presproint{k'} (l_{k'}^{2} \proba{k=\inf\{P_{k'}^{B} - \inf\{P_{k'}^{B}\}\}}\\
& + l_{k'} (r_{k'}+1) \proba{k=\inf\{P_{k'}^{B}\}}\\
& + r_{k'} (l_{k'}+1) \proba{k=\sup\{P_{k'}^{S}\}}\\
& + r_{k'}^{2} \proba{k=\sup\{P_{k'}^{S} - \sup\{P_{k'}^{S}\}\}})
\end{align*}

\begin{align*}
\proba{\op{\formins}{k'} \leadsto \opcas{\nodeint{k}}}=&
(1-\presproint{k'}) (l_{k'}^{2} \proba{k=\inf\{P_{k'}^{B} - \inf\{P_{k'}^{B}\}\}}\\
& + l_{k'} r_{k'} \proba{k=\inf\{P_{k'}^{B}\}}\\
& + r_{k'} l_{k'} \proba{k=\sup\{P_{k'}^{S}\}}\\
& + r_{k'}^{2} \proba{k=\sup\{P_{k'}^{S} - \sup\{P_{k'}^{S}\}\}})
\end{align*}

Lastly, we split the \castxts events to the virtual nodes. 
\castxts events can happen at the internal nodes only when
they are in the last two levels of the tree (or similarly when the
size of the subtree that is rooted at the concerned internal node is
in the interval $[1,3]$). We required the average behaviour to be
valid and set $\rateCasEI{k}{x} = \presproint{k} \rateCasInt{k} / (\presprointv{k}{1} + \presprointv{k}{2} + \presprointv{k}{3}), \forall x \in \{1,2,3\}$.
%
%
For the cases where the operation key selection follows a zipf
distribution, there exist a small region of the tree that the most
operations concentrate.  
The update operations
concentrate to that region so that the nodes are expected to change
levels frequently. This means that the impact of invalidation 
recovery factor can be seen while the node is at an level. 
For this impacting
factor, for zipf distribution, we split the events to virtual nodes 
evenly, $\forall h, \rateCasEI{k}{h}=\rateCasInt{k}$.


}

\section{Experimental Evaluation\pp{\vsme}}
\label{sec:expMain}

\pp{
\begin{figure}

\begin{minipage}{.32\linewidth}
\includegraphics[width=\textwidth]{./results/pdfs/ht-1_uni-small.pdf}
\caption{\dshttit with load factor 2\label{fig:ht1uni-small}}
\end{minipage}%
\begin{minipage}{.32\linewidth}
\includegraphics[width=\textwidth]{./results/pdfs/skl_uni-small.pdf}
\caption{\dssltit\label{fig:skluni-small}}
\end{minipage}%
\begin{minipage}{.32\linewidth}
\includegraphics[width=\textwidth]{./results/pdfs/bst_uni-small.pdf}
\caption{\dsbttit\label{fig:bstuni-small}}
\end{minipage}%
\end{figure}
}

We validate our model through a set of well-known \lf \sds designs,
mentioned in the previous section. We stress the model with various
access patterns and number of threads to cover a considerable amount
of scenarios where the \dss could be exploited. 
For the key selection process, we vary the key ranges and the
distribution: from uniform (\ie the probability of targeting any key
is constant for each operation) to zipf (with $\alpha=1.1$ and the
probability to target a key decreases with the value of the key).
Regarding the operation types, we start with various balanced update
ratios, \ie such that the ratio of \ins (among all operations) equals
the ratio of \del. Then, we also consider asymmetric cases where the
ratio of \ins and \del operations are not equal,
which changes the expected size of the \shds.
\pp{Due to space constraints, most of the experimental evaluation is 
presented in~\cite{long-version}.\vsbi{}}

\subsection{Setting\pp\vssm}

We have conducted experiments on an Intel ccNUMA workstation
system. The system is composed of two sockets, each containing eight
physical cores. The system is equipped with Intel Xeon E5-2687W v2
CPUs. Threads are pinned to separate cores. One can observe the
performance change when number of threads exceeds $8$, which activates
the second socket.

In all the figures, y-axis provides the throughput, while the number
of threads is represented on x-axis. The dots provide the results of
the experiments and the lines provide the estimates of our
framework. The key range of the \ds is given at the top of
the figures and the percentage of update operations are color coded.

We instantiate all the algorithm and architecture related latencies, 
following the methodologies described in~\cite{intel-bench,measure}.
In line with these studies, we observed that the latencies of \latcas
and \latrec are based on thread placement. We distinguish two
different costs for \latcas according to the number of active
sockets.
Similarly, given a thread accessing to a node \node{i}, the recovery
latency is low (resp. high), denoted by \latrecSl (resp. \latrecSh),
if the modification has been performed by a thread that is pinned to
the same (resp. another) socket. Before the execution, we measure both
\latrecSl and \latrecSh, and instantiate \latrec with the average
recovery latency, computed in the following way for a two-socket chip.
For $s \in \{ 1, 2\}$, we denote by \nbthS{s} the number of threads
that are pinned to socket numbered $s$. By taking into account all combinations, we have
$ \latrec =  ( \nbthS{1} (\nbthS{1} \latrecSl + \nbthS{2} \latrecSh) + \nbthS{2} (\nbthS{2} \latrecSl + \nbthS{1} \latrecSh))/ \nbth^2$.
Since $\nbth = \nbthS{1} + \nbthS{2}$, we obtain 
$\latrec = \latrecSl + 2(\nbthS{1}/\nbth)(1 - \nbthS{1}/\nbth)(\latrecSh - \latrecSl)$.


For the \ds implementation, we have used ASCYLIB
library~\cite{trigonakis} that is coupled with an epoch based memory
management mechanism which introduces negligible latency.\pp\vsmb

\subsection{\sdstits\pp\vsme}

\pr{
In Figure~\ref{fig:ht1uni-small}, \ref{fig:skluni-small}
and~\ref{fig:bstuni-small}, we provide the results for the \dsht,
\dssl and \dsbt where the key selection is done with the uniform
distribution. We see that our framework provides satisfactory
estimations for all cases. One can see that the performance drops as
the update rate increases, due to the impact of \castxts related
factors. This impact magnifies with the activation of the second
socket (more than 8 threads) since the event becomes more costly.
When there is no update operation, the performance scales linearly
with the number of threads. As a results of disjoint access
parallelism, we observe a similar behaviour for the cases with
updates. See~\cite{long-version} for a more comprehensive
set of experiments.\vsbi
}
{


\subsubsection{\dslltit}

Figures~\ref{fig:lluni}, \ref{fig:llzipf} and~\ref{fig:llasym}
illustrates the results for the \lf \dsll, for various scenarios that
are described before (see~\ref{sec:expMain}).  For the majority of the
cases, our estimates look reasonable except the cases where the cache
miss ratios are underestimated due to the limitations of the
independent reference assumption. The assumption in the Independent
Reference Model is that the event at the different nodes are
independent Poisson Processes. A \dsll operation reveals a high
degree of spatial locality, implying that the Poisson Processes for
the different nodes are indeed dependent. This inaccuracy illustrates
indeed the importance of the accurate estimations for the event
latencies that are needed to capture the practical performance.

\begin{figure}[!htb]
\begin{center}
\includegraphics[width=\textwidth]{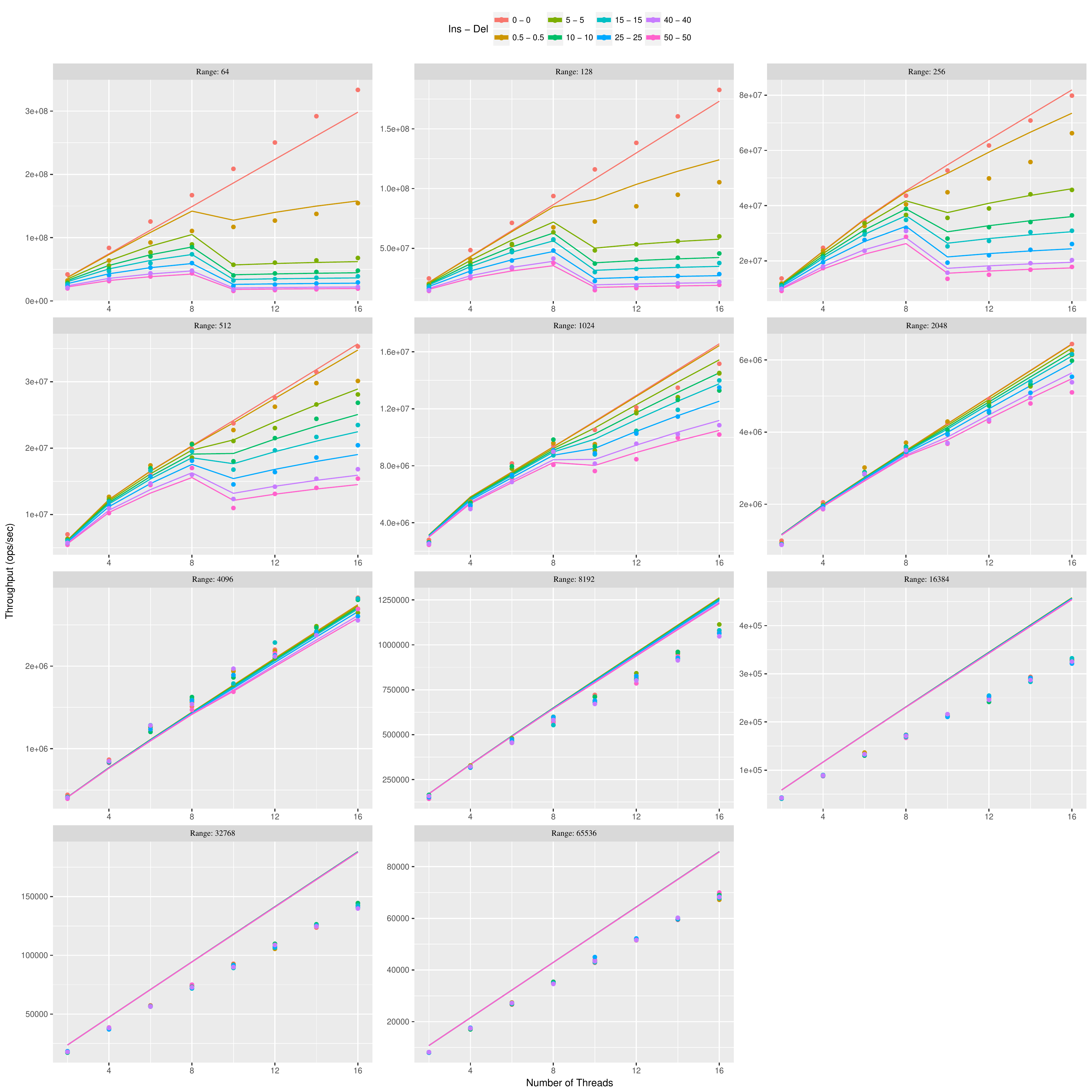}
\end{center}
\caption{LL Uniform distribution for key selection\label{fig:lluni}}
\end{figure}

\begin{figure}[!htb]
\begin{center}
\includegraphics[width=\textwidth]{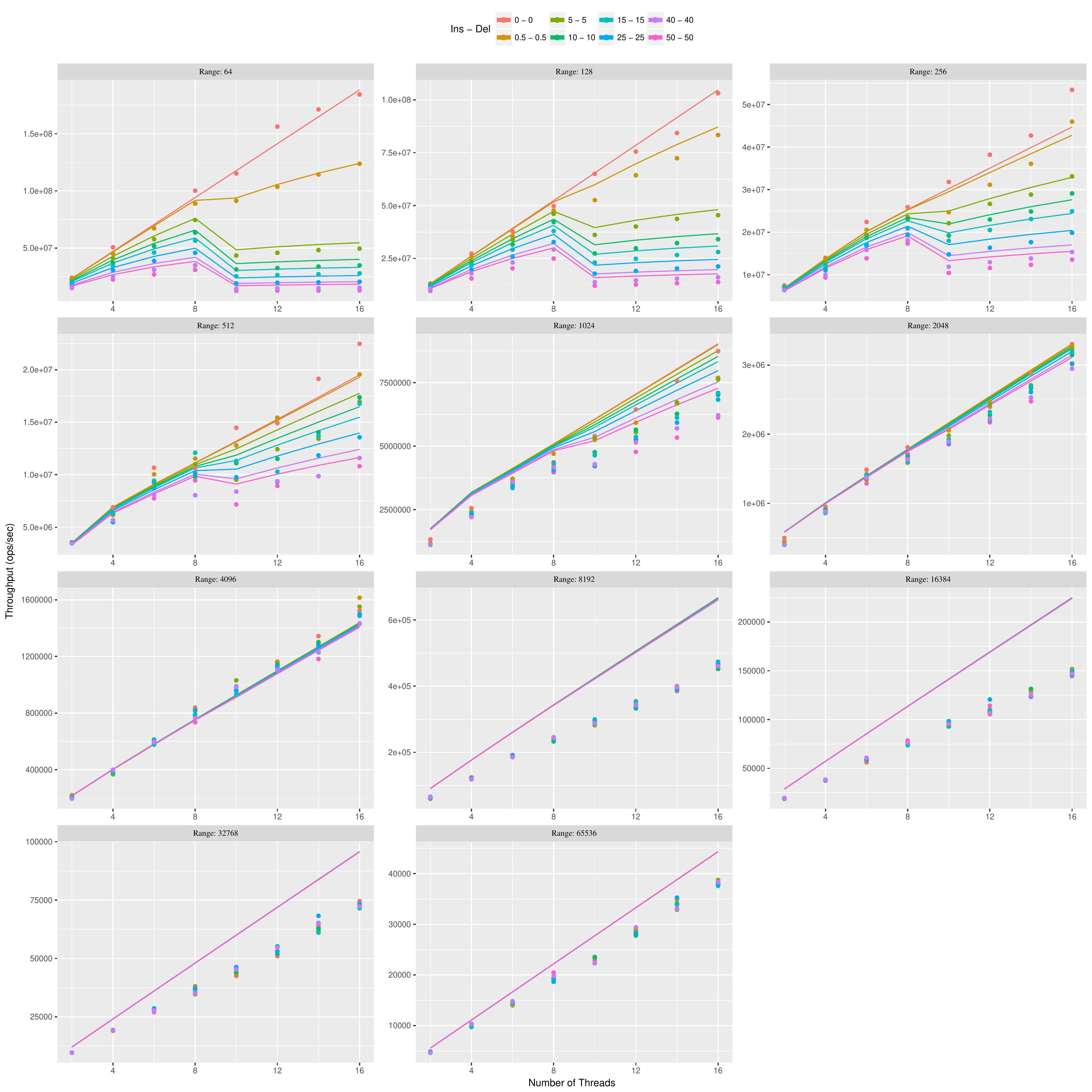}
\end{center}
\caption{LL Zipf distribution for key selection\label{fig:llzipf}}
\end{figure}

\begin{figure}[!htb]
\begin{center}
\includegraphics[width=\textwidth]{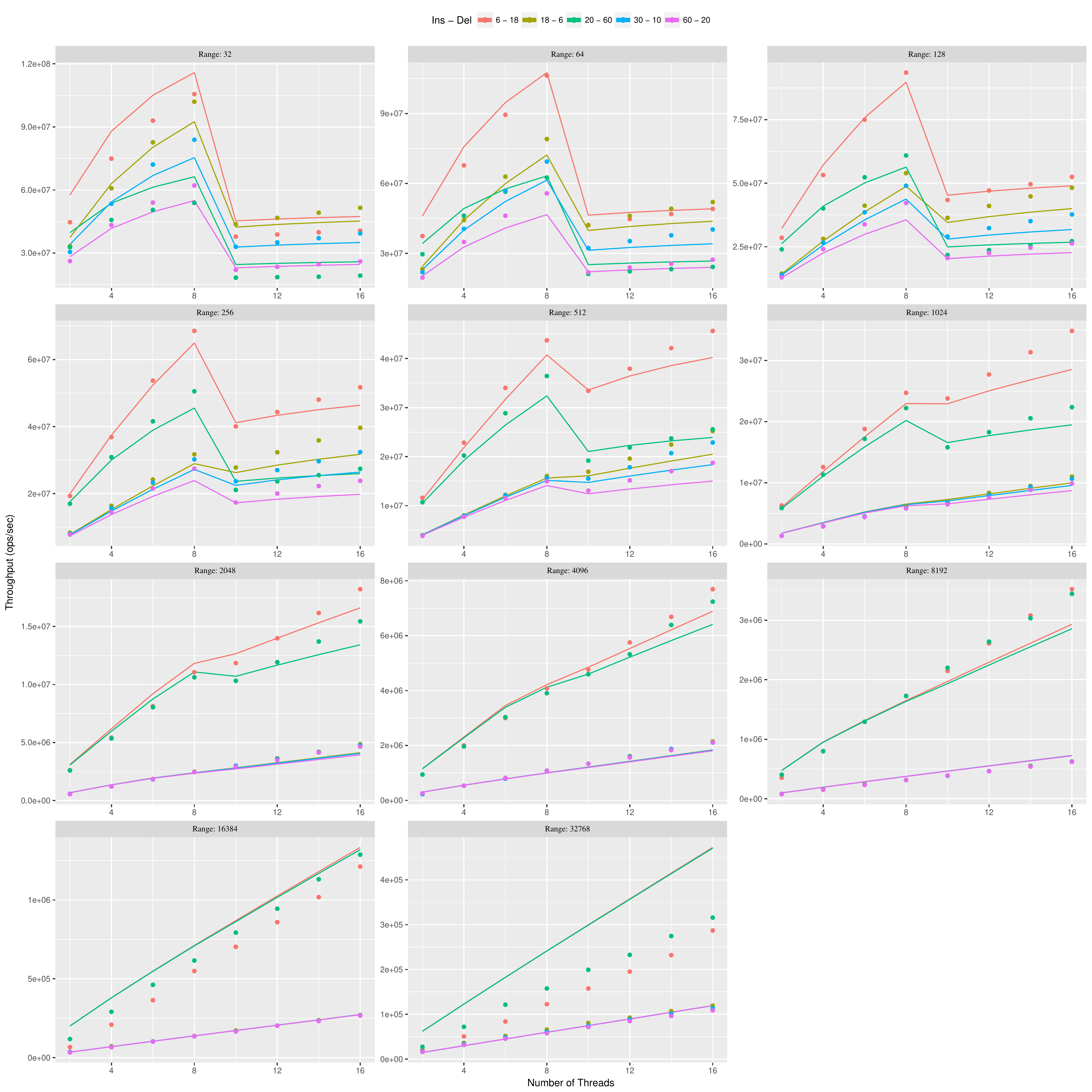}
\end{center}
\caption{LL asymmetric update rates, uniform distribution for key selection\label{fig:llasym}}
\end{figure}

\clearpage
\subsubsection{\dshttit}

Figure~\ref{fig:ht1uni}, \ref{fig:ht2uni} and~\ref{fig:ht4uni}
illustrates the results for the \lf \dsht with different load factor
values (number of slots per bucket) where the key selection process is
initiated with uniform distribution. Figure~\ref{fig:ht1zipf} shows
the results for a case where the selection process follows zipf
distribution.  Lastly, Figure~\ref{fig:ht2asym} reveals the results
for asymmetric delete and insert operation ratios where the key
selection is done with uniform distribution.  For the \dsht, our
estimates are able to capture the real behavior almost for all cases
with satisfactory precision.


\begin{figure}[!htb]
\begin{center}
\includegraphics[width=\textwidth]{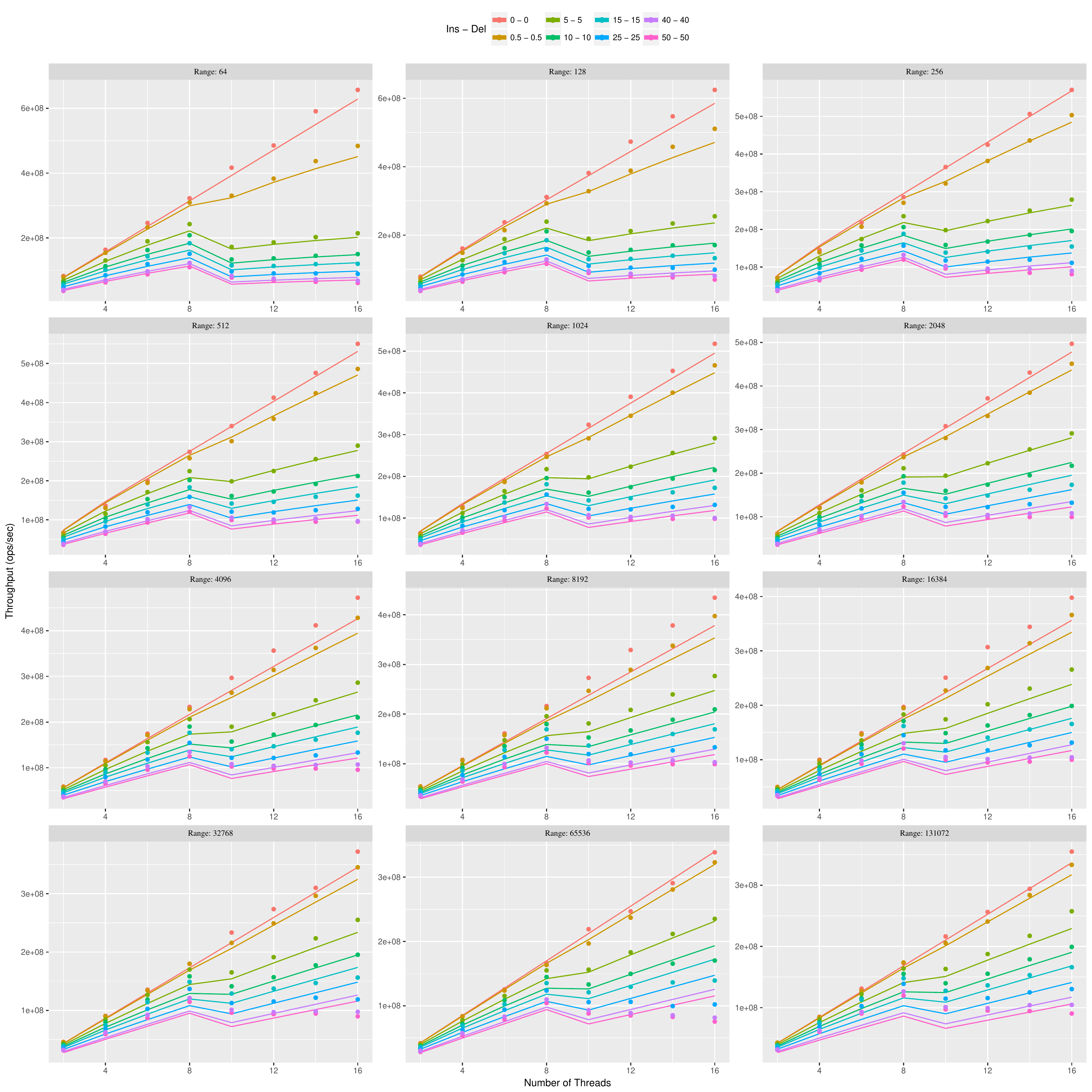}
\end{center}
\caption{HT Uniform distribution for key selection, with load factor=2\label{fig:ht1uni}}
\end{figure}

\begin{figure}[!htb]
\begin{center}
\includegraphics[width=\textwidth]{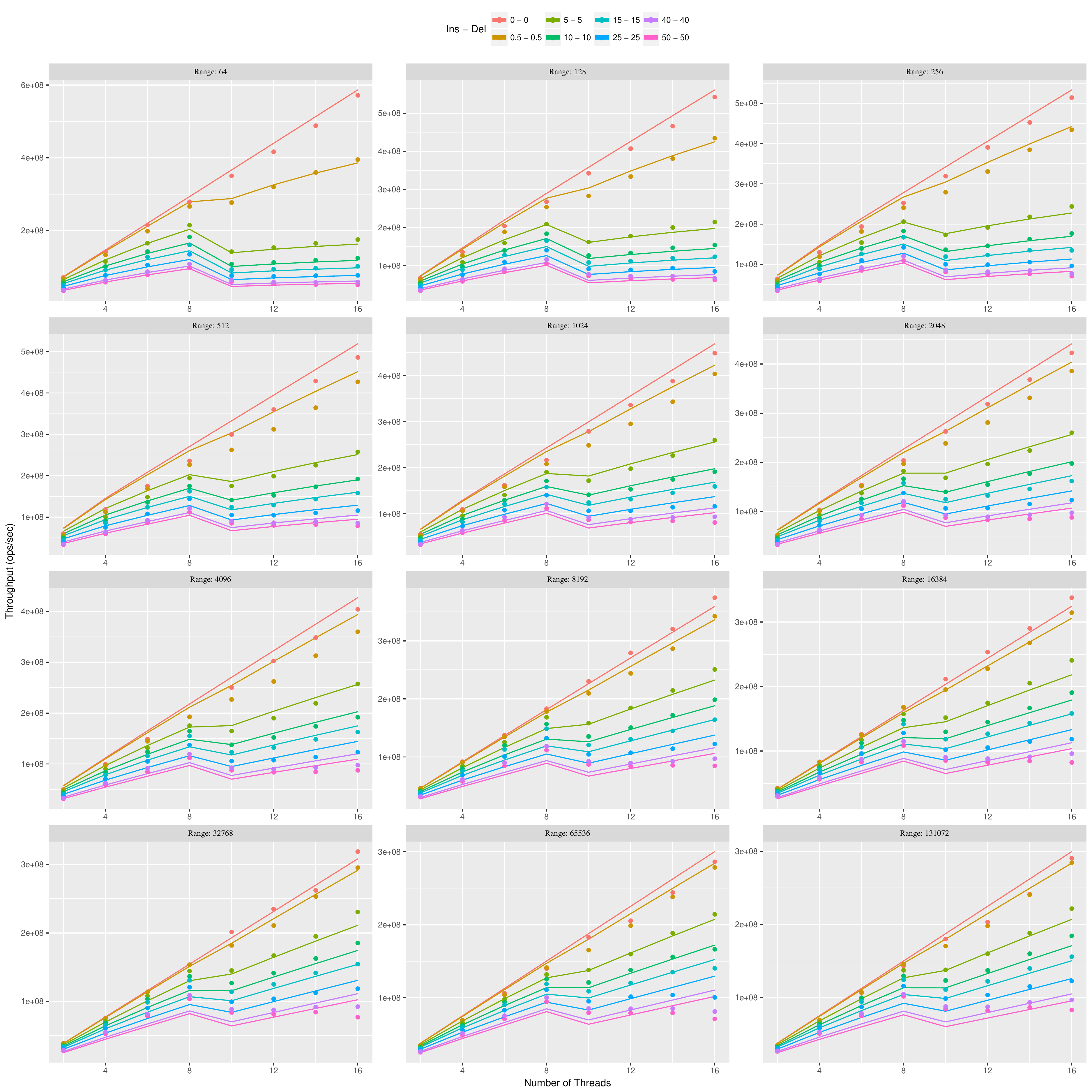}
\end{center}
\caption{HT Uniform distribution for key selection, with load factor=4\label{fig:ht2uni}}
\end{figure}

\begin{figure}[!htb]
\begin{center}
\includegraphics[width=\textwidth]{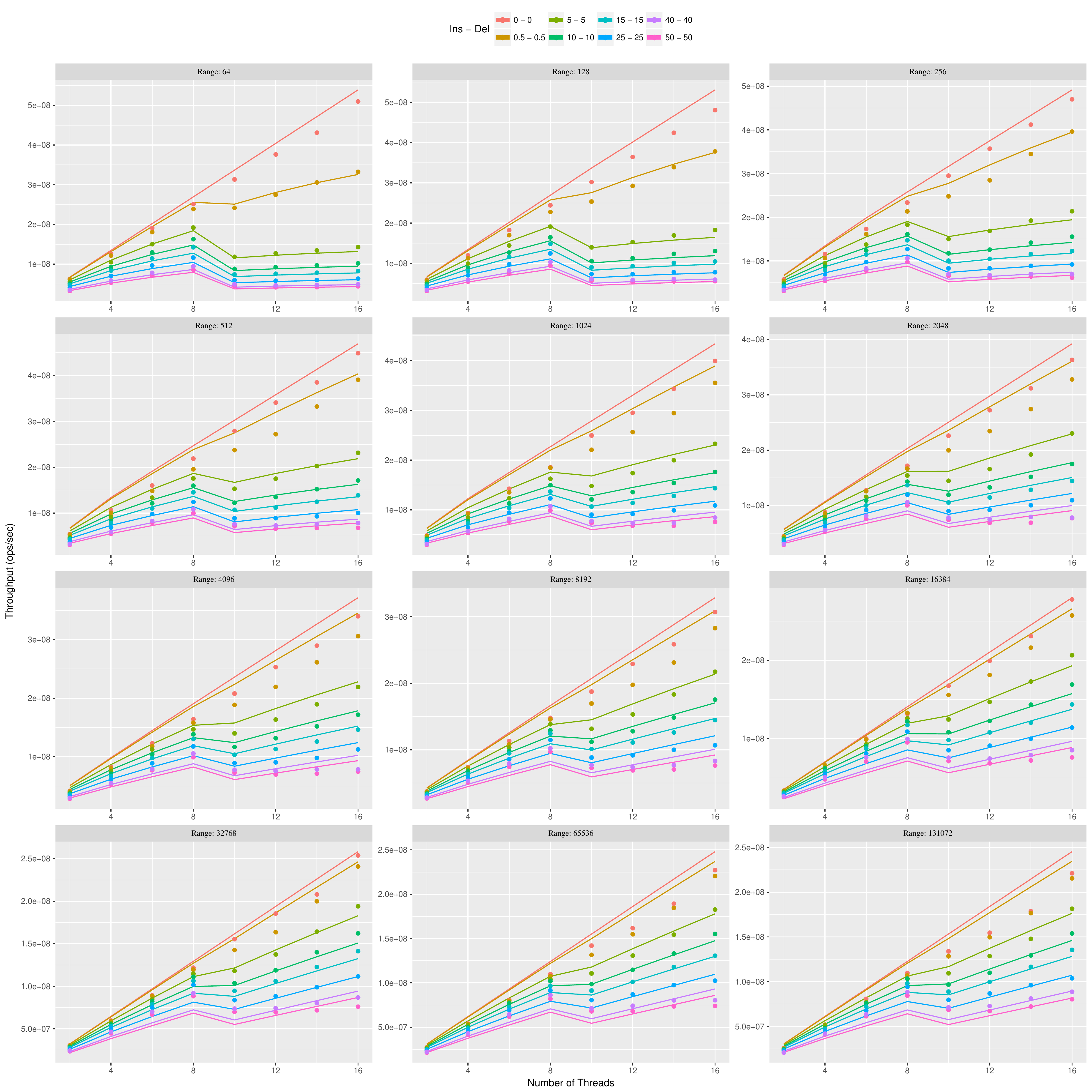}
\end{center}
\caption{HT Uniform distribution for key selection, with load factor=8\label{fig:ht4uni}}
\end{figure}

\begin{figure}[!htb]
\begin{center}
\includegraphics[width=\textwidth]{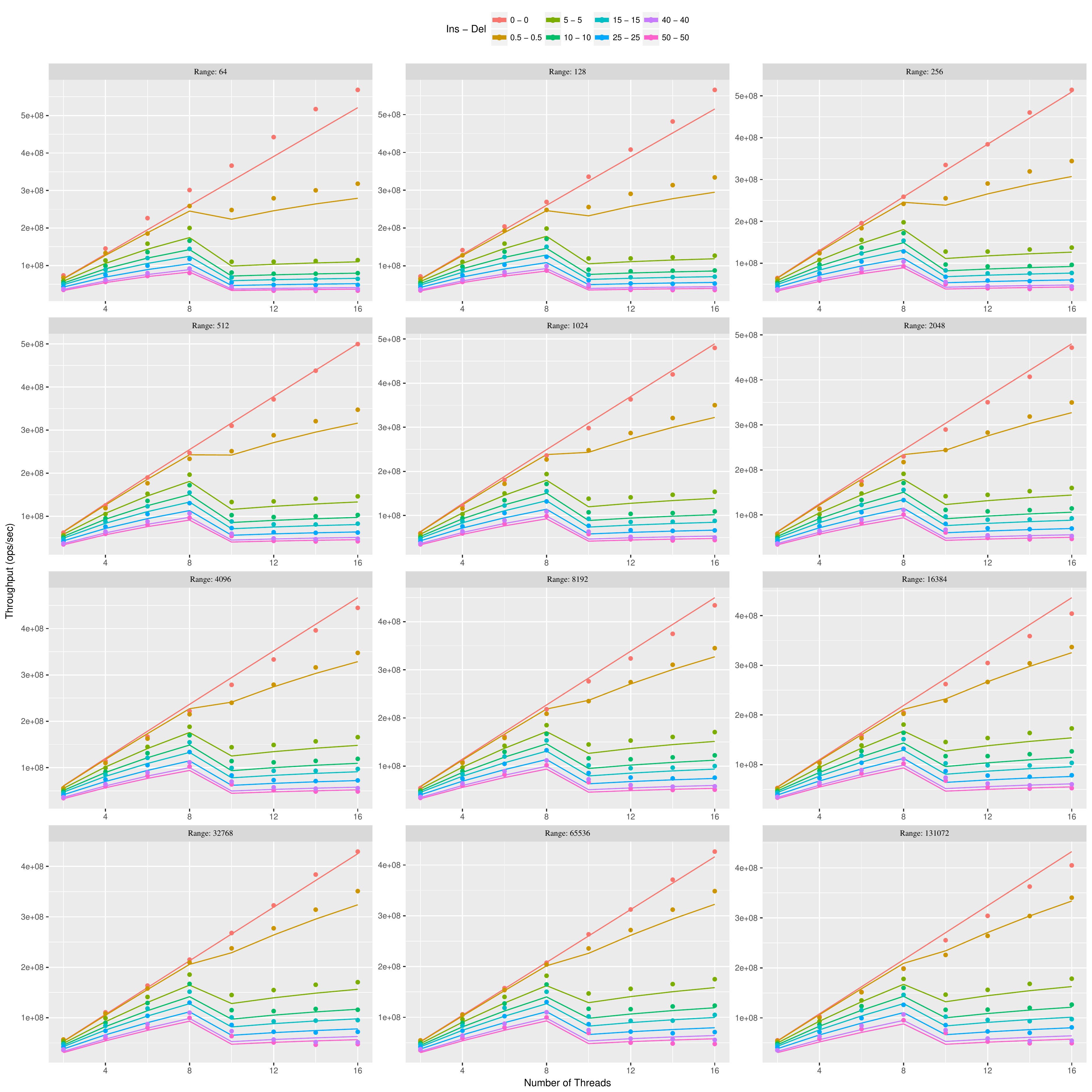}
\end{center}
\caption{HT Zipf distribution for key selection, with load factor=2\label{fig:ht1zipf}}
\end{figure}

\begin{figure}[!htb]
\begin{center}
\includegraphics[width=\textwidth]{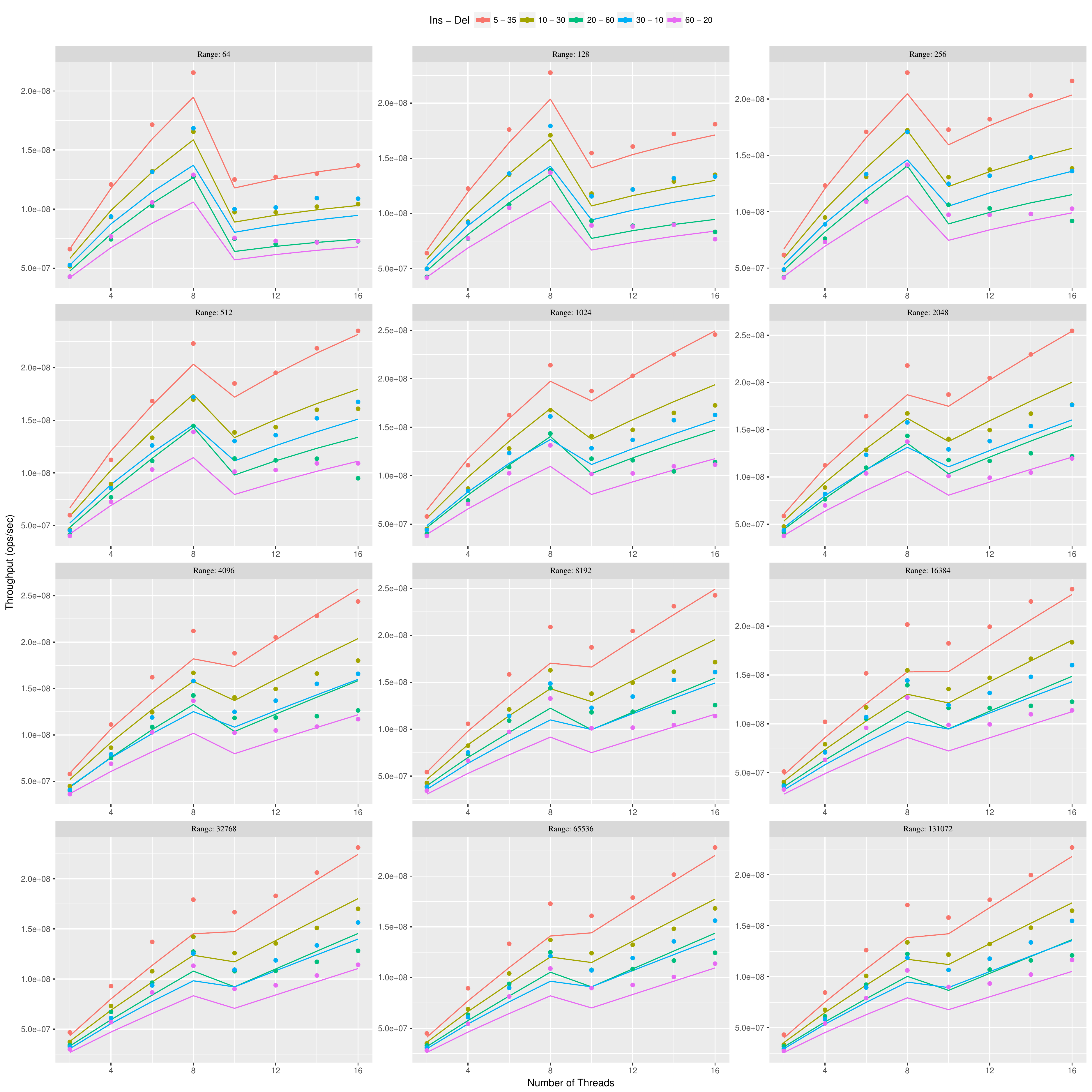}
\end{center}
\caption{HT asymmetric update operations, Uniform distribution for key selection, with load factor=4\label{fig:ht2asym}}
\end{figure}

\clearpage

\subsubsection{\dssltit}

Figure~\ref{fig:skluni}, \ref{fig:sklzipf} and~\ref{fig:sklasym}
illustrates the results for the \lf \dssl, for various scenarios that
are described before (see~\ref{sec:expMain}), where the estimations
often closely follow the real behavior.  In Figure~\ref{fig:sklasym},
we observe that our estimation show some deviation from the real
behavior, for the cases where key range is small and \del ratio is
higher than \ins.  For such cases, the expected size of \sds tends to
be very small which might lead to inaccuracies.

\begin{figure}[!ht]
\begin{center}
\includegraphics[width=\textwidth]{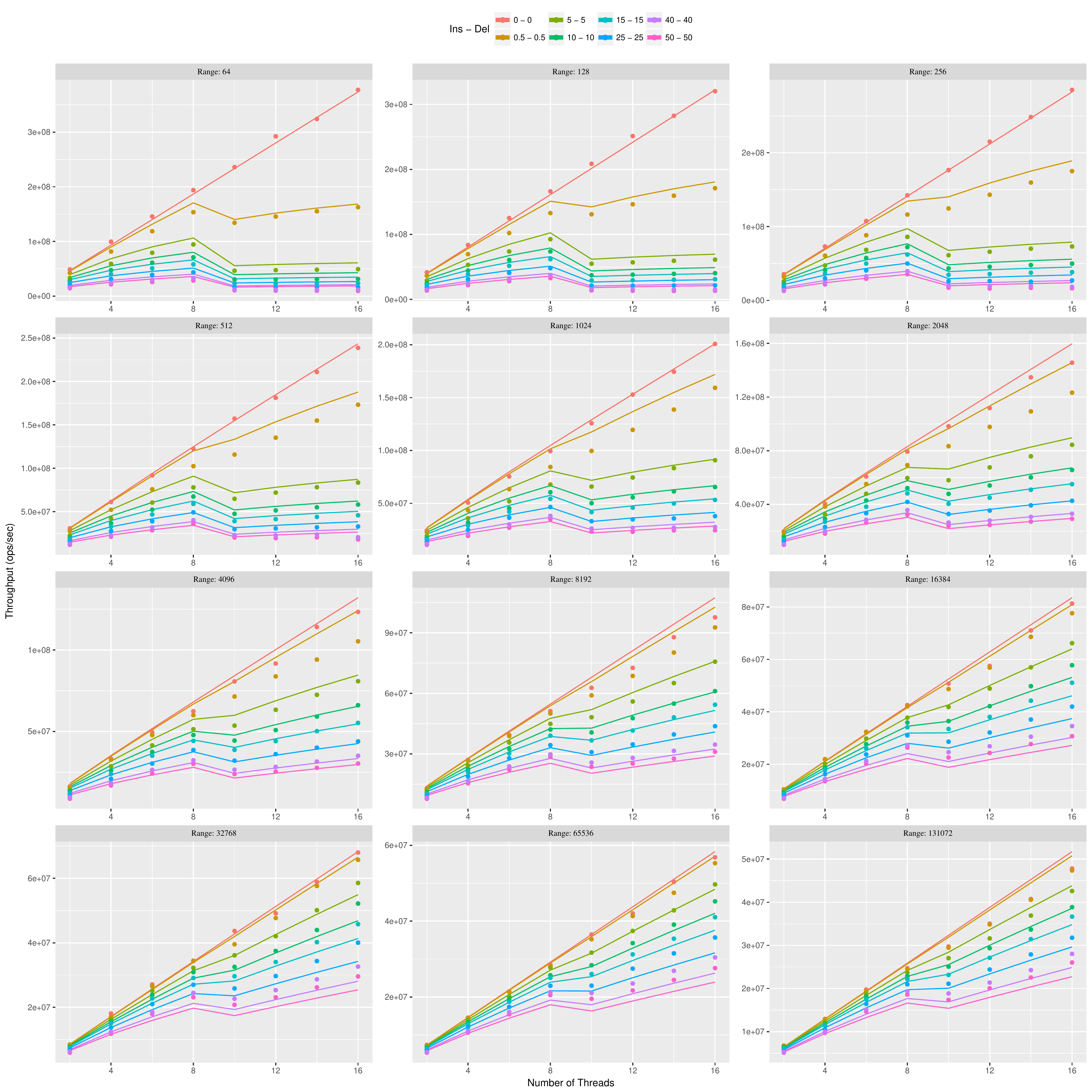}
\end{center}
\caption{Skiplist Uniform distribution for key selection\label{fig:skluni}}
\end{figure}

\begin{figure}[!ht]
\begin{center}
\includegraphics[width=\textwidth]{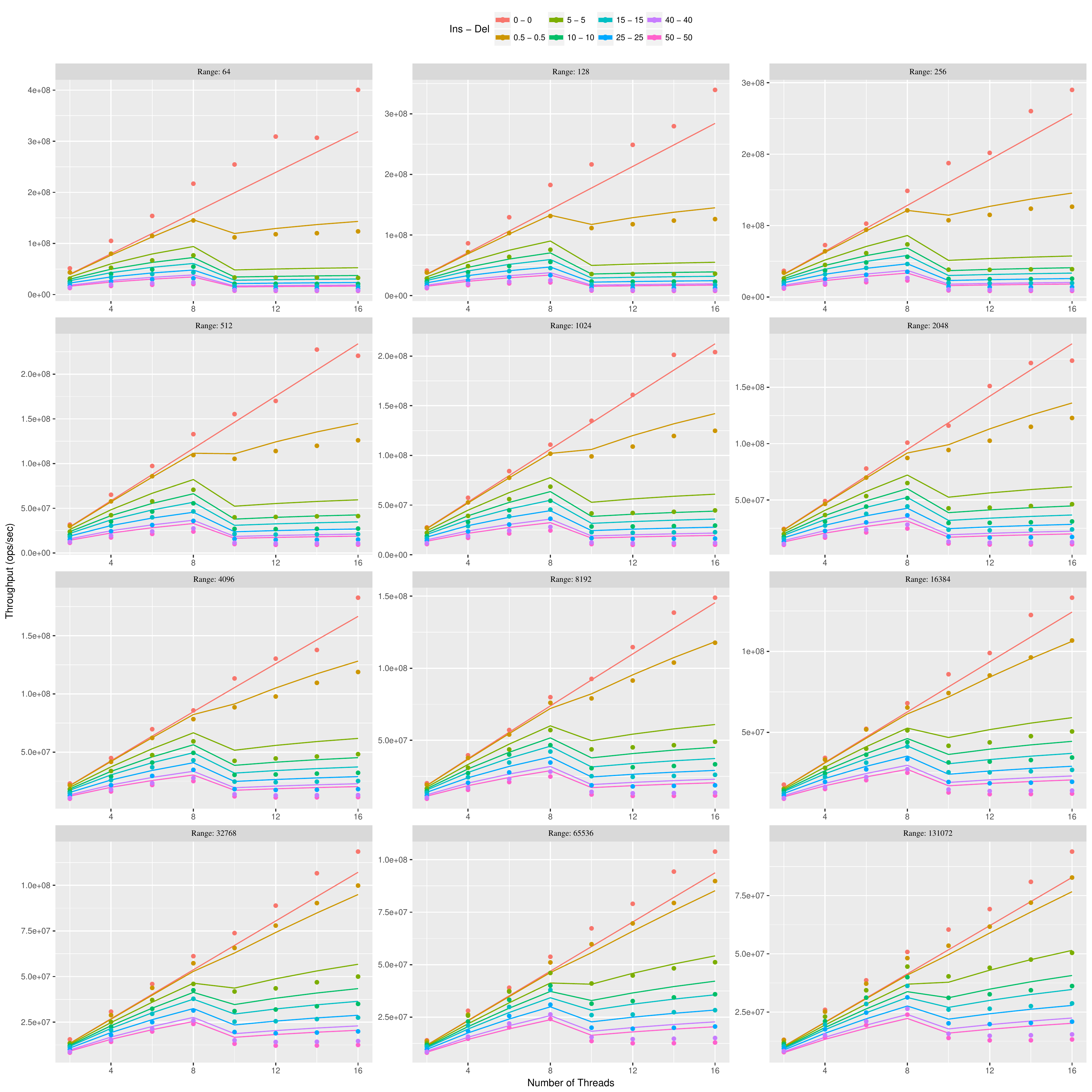}
\end{center}
\caption{Skiplist Zipf distribution for key selection\label{fig:sklzipf}}
\end{figure}

\begin{figure}[!ht]
\begin{center}
\includegraphics[width=\textwidth]{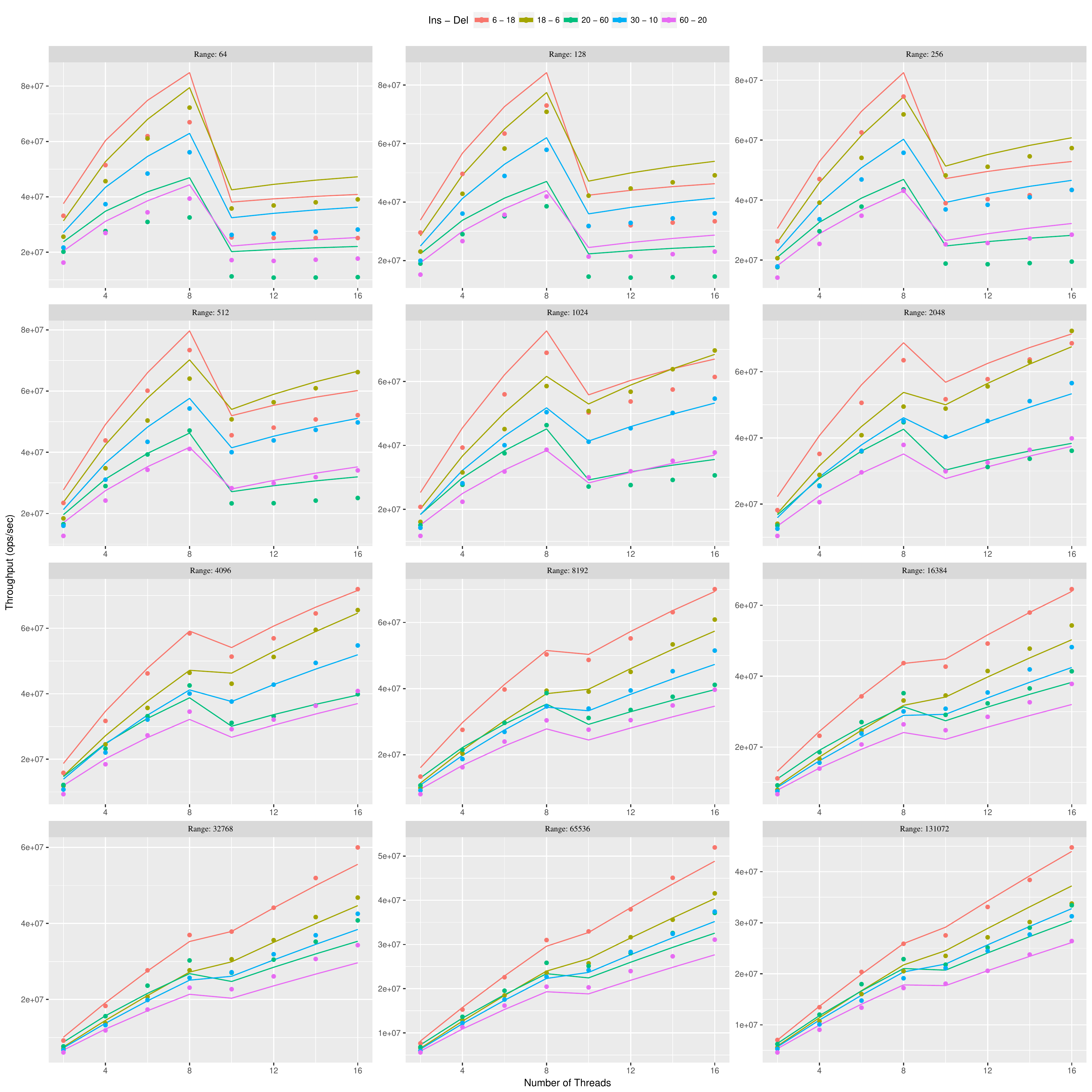}
\end{center}
\caption{Skiplist asymmetric update rates, uniform distribution for key selection\label{fig:sklasym}}
\end{figure}

\clearpage
\subsubsection{\dsbttit}

Figure~\ref{fig:bstuni}, \ref{fig:bstzipf} and~\ref{fig:bstasym}
illustrates the results for the \dsbt, for various scenarios that are
described before (see~\ref{sec:expMain}).  Here, we observe that our
estimations often closely follow the real behaviour.

\begin{figure}[!ht]
\begin{center}
\includegraphics[width=\textwidth]{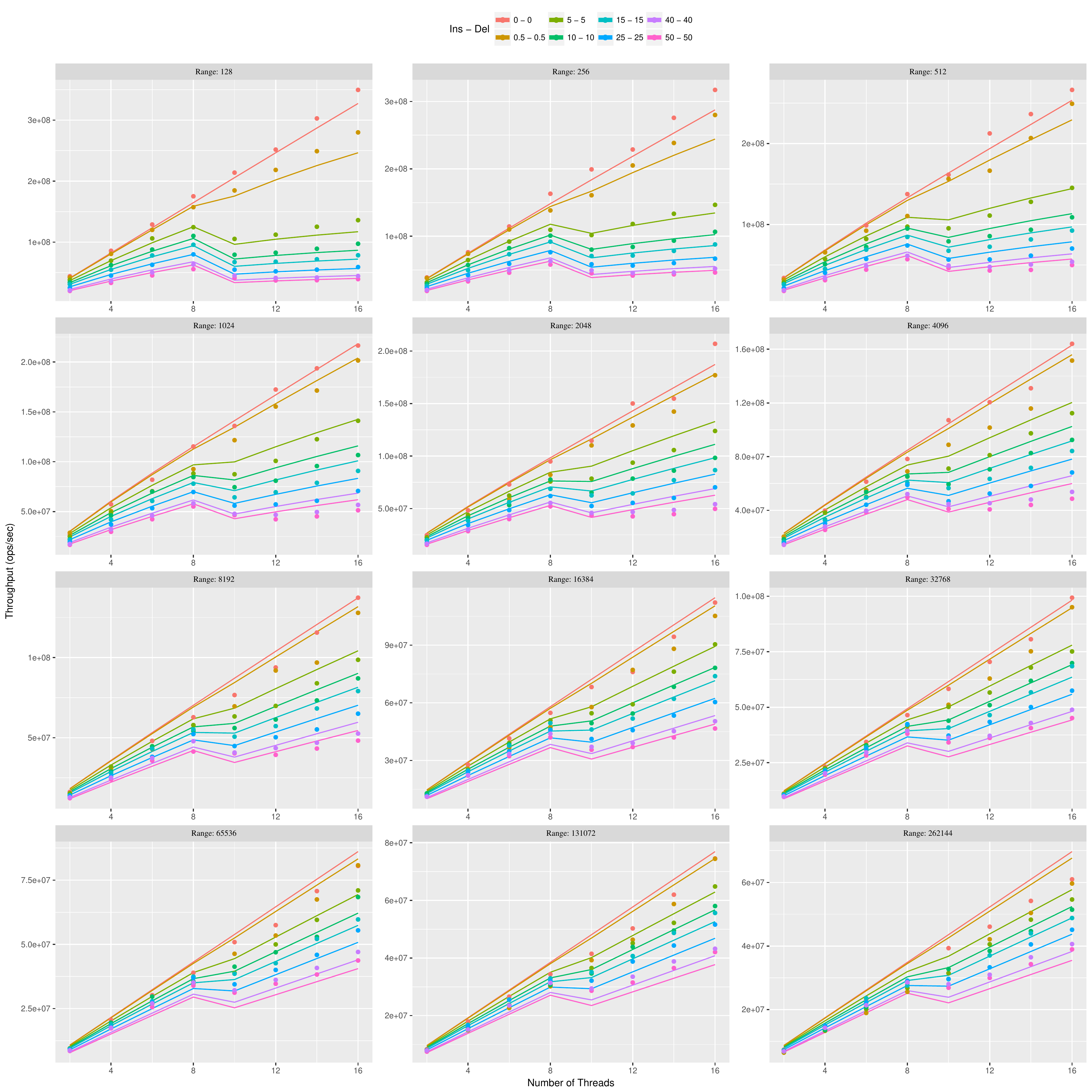}
\end{center}
\caption{BST Uniform distribution for key selection\label{fig:bstuni}}
\end{figure}

\begin{figure}[!ht]
\begin{center}
\includegraphics[width=\textwidth]{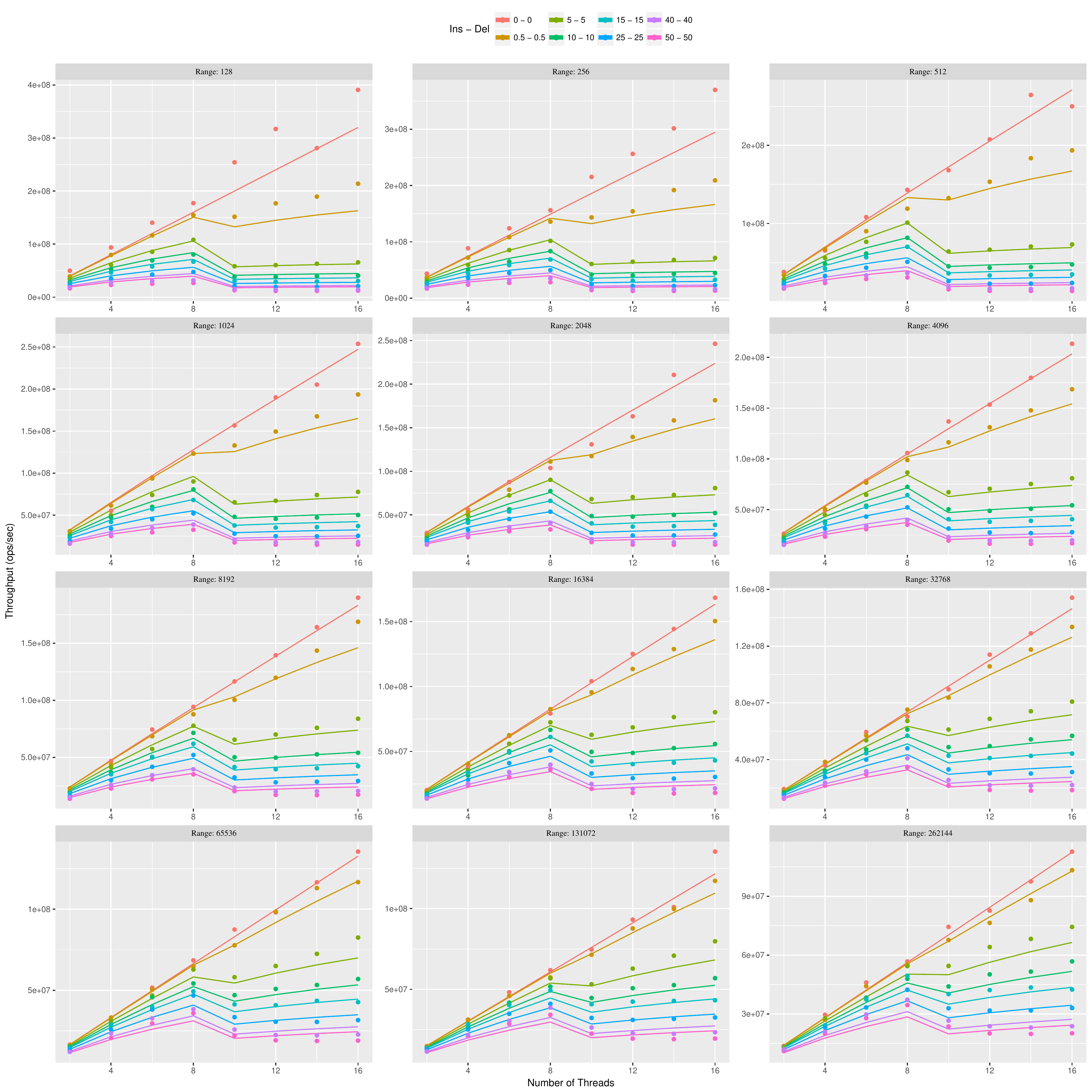}
\end{center}
\caption{BST Zipf distribution for key selection\label{fig:bstzipf}}
\end{figure}

\begin{figure}[!ht]
\begin{center}
\includegraphics[width=\textwidth]{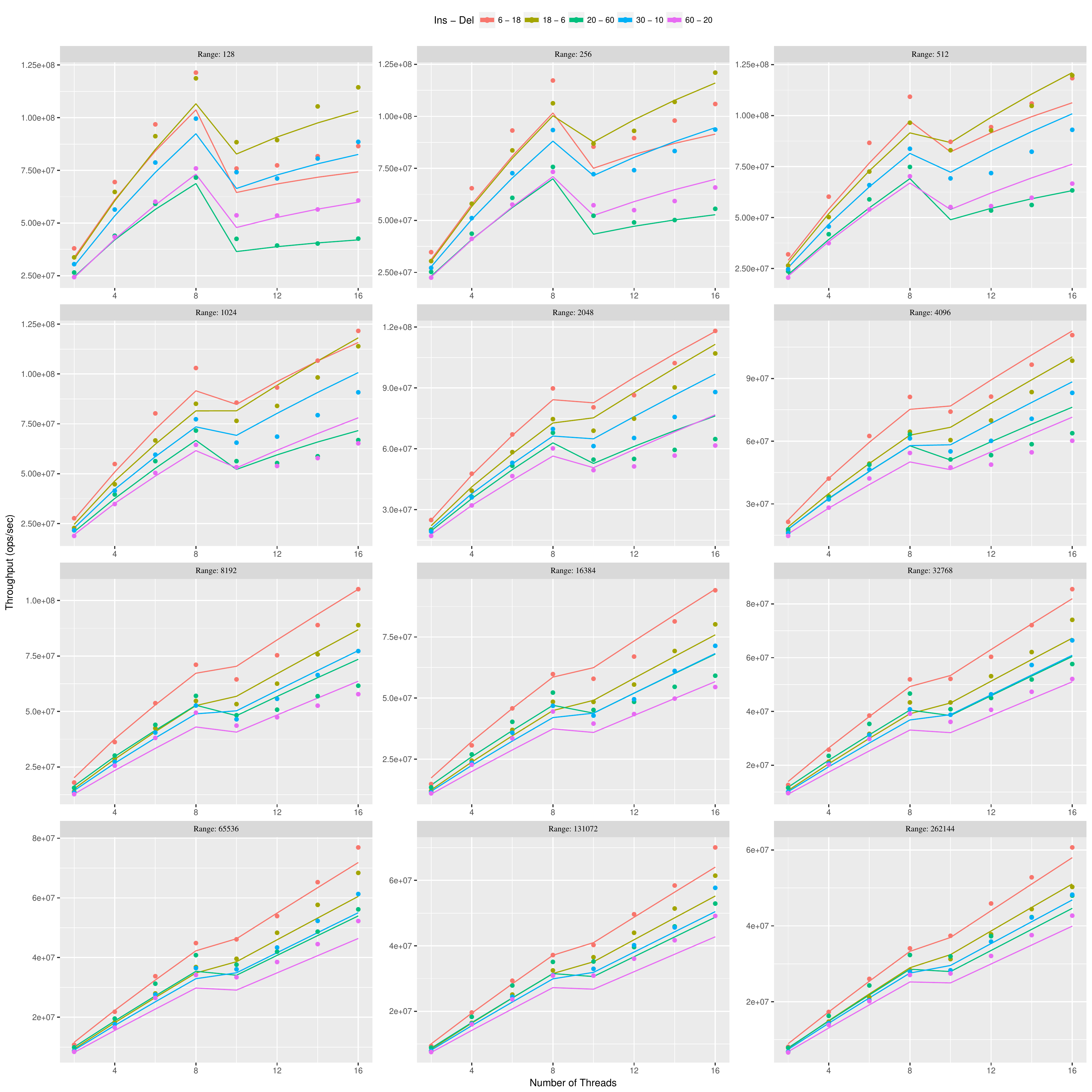}
\end{center}
\caption{BST asymmetric update rates, uniform distribution for key selection\label{fig:bstasym}}
\end{figure}

\clearpage

}

%

\section{Applications: to Pad or not to Pad\pp\vssm}
\label{sec:padMain}

\newcommand{\rateCasAddi}[1]{\ema{\lambda_{#1}^{cas, addi}}}
\newcommand{\rateReadAddi}[1]{\ema{\lambda_{#1}^{read, addi}}}


In a non-padded (packed) configuration,
multiple nodes are packed together into a single cacheline. This
implies that a modification done at a node, could lead to a
coherence cache miss in the traversal of the other nodes.
It is often referred as {\it false sharing}.
On the other
hand, the packed configurations benefit from their compact
representation by reducing the capacity misses.

Until now, we have assumed that the nodes are padded. Here, we extend
the framework to estimate the performance of a packed configuration to
facilitate the tuning process.
In such a setting, where the nodes are inserted and deleted
repeatedly, \node{i} can be alone in its cacheline with the old
versions of a set of nodes that are not present any more in the
\ds. Alternatively, it might be mapped to the same cacheline with some
number of active nodes that are present in the \sds and they all
together contribute to the event rates that are originating from the
same cacheline.

\newcommand{\nbs}{\ema{S}}
\newcommand{\pgsi}{\ema{\mathit{pageSize}}}


Firstly, we assume that at most two nodes can be packed to a cacheline
(which is the case for the \dss that we consider) and we denote the
total number of slots for the node allocations with $\nbs = 2 \nbpg
\pgsi/cacheLineSize$ (recall that \nbpg is the number of pages that are used
by the \shds). We assume that the nodes are assigned uniformly to the slots; 
given that \node{i} and \node{j} are present in the \shds, \node{j} is
mapped to the same cacheline as \node{i} with probability: $1/(\nbs-1)$.
With the linearity of expectation, the expected additional event rate
for the cacheline that \node{i} is mapped to can be given by the sum
of event rates originating from different nodes.
\rateRead{i} and \rateCas{i} provides the event rates for \node{i},
and we introduce an additive factor to represent the average event
rate contributions of other nodes to the cacheline of \node{i}:
\rateReadAddi{i} for \Read events, and \rateCasAddi{i} for \castxts
events.
\node{j} contribute to the \Read event rates with \rateRead{j} if
\node{j} and \node{i} are assigned to the same cacheline, which
happens with probability $\prespro{j}/(\nbs-1)$. Then, we have:
$\rateReadAddi{i} = \sum_{j=1, j\neq i}^{\nbn} \rateRead{j}
(\frac{\prespro{j}}{\nbs-1})$ and $\rateCasAddi{i} = \sum_{j=1, j\neq
  i}^{\nbn} \rateCas{j} (\frac{\prespro{j}}{\nbs-1})$.

\pp{
\begin{wrapfigure}[14]{l}{.5\textwidth}
\includegraphics[width=.5\textwidth]{./results/pdfs/ht1pack_uni-small.pdf}
\caption{Packed nodes (dots: measure, solid: estimate) and padded (dashed: measure) for HT}
\label{fig:htpacksmall}
\end{wrapfigure}
}

\pr{\noindent{}In~\cite{long-version}, we show how to integrate these
  additive factors to each performance impacting factor and inspect
  them. The packing
  has a positive impact on
  cache misses. A less expected outcome is that we
  do not observe any negative trace of the packing, even for the cases
  under high update rates, because the additive contributions scale
  and cancel each other in the ratio of the recovery factor.
  We depict the results in Figure~\ref{fig:htpacksmall} for \dshts with different configurations:
  padded nodes (dashed lines), packed nodes
(dots) and our estimations for the packed nodes (lines).\vsbi
}
{


With the node packing, we obtain additive components for \castxts and
\Read events.
\cut{
given that \node{i} and \node{j} are present in the \shds, \node{j} is
mapped to the same cacheline as \node{i} with probability: $1/(M-1)$.
, where probability of \node{i} to be with
a node ${{{M-2}\choose{k-2}}/ {{M-1}\choose{k-1}}} = \frac{k-1}{M-1}$
and $\frac{1}{k-1}$ is the probability of \node{j} is the one.
}
Now, we show the
integration of these additive components into the process.

\subsubsection{Cache Misses}

To begin with, packing would have a positive impact on the cache misses as it
would increase the characteristic time ($T$) of the cache, that is the
duration for $C$ unique cacheline references. To recall, \node{i} could
contribute to this $C$ references only if $\node{i} \in \DS$ and we have embedded this effect
into the process by introducing the random variable $P_i$ (see~\ref{sec:cachemiss}).
With the packing, this contribution becomes less probable, as the contribution would occur
only if the reference to \node{i} occurs before the references to the other node that is
mapped to the same cacheline with \node{i}. Otherwise, the reference to \node{i} would be
ineffective for the characteristic time. To recall, the characteristic time is the solution
of the following equation:

\[X^{cache}(t)=\sum_{j=1}^{N} P^{pack}_i\mathbf{1}_{0<\obj{j}\leq t}\]

where $P^{pack}_i$ is the variable that we modify in the process,

\[
P^{pack}_i=
\begin{cases}
p_i (\rateRead{i} / (\rateRead{i} + \rateReadAddi{i})) , & \text{if } P^{pack}_i=1\\
1-p_i (\rateRead{i} / (\rateRead{i} + \rateReadAddi{i})), & \text{if } P^{pack}_i=0
\end{cases}
\]


Having obtained the characteristic time, we involve the additive factor to estimate the
cache miss rate of \node{i}. This is because a reference leads to a cache miss (in a cache
of size $C$) only if the previous $C$ cacheline references do not include the cacheline
that \node{i} is mapped to.

\[ Hit_i^{cache} = 1 - e^{(-(\rateRead{i} + \rateReadAddi{i}))/\nbth) T} \]

\subsubsection{Page Misses}

Secondly, packing can improve the TLB cache hit ratios. This simply happens because it
reduces the total number of pages that the \sds spans. To recall, the total number of pages
is a parameter of the process that computes the expected latency for the impacting factor
($Hit_i^{tlb}$). Packing do not influence the process, so we just need to update the
value of the parameter.

\subsubsection{CAS Execution}

On the downside, packing is expected to reduce the performance through
the \castxts related impacting factors. To recall, \casrec{i}
represents the expected latency per traversal at \node{i} for
executing \castxts instructions targeted to \node{i}. This factor is
proportional to the throughput, and packing do not change the
probability of executing a \castxts at \node{i} while traversing it. So,
packing does not have a direct impact on this component.

\subsubsection{Invalidation Recovery}

The most important performance impacting \castxts related factor is
the invalidation recovery.  For each traversal of \node{i}, there
exist a possibility to pay for a coherence cache miss due to the
previous \castxts executions at the cacheline, that \node{i} is mapped
to.  To compute the probability of a coherence miss, one needs to
consider the previous events on the cacheline. The traversal (by a
thread at \node{i}) would not experience the coherence miss if the
previous traversal (on the cacheline that \node{i} is mapped to) of
the same thread is not followed by \castxts event of another
thread. Thus, we consider the additive factor for both type of events
and modify the process as follows:

\[ \proba{\text{Coherence Miss on traversal of \node{i}}} =
\frac{(\rateCas{i} +\rateCasAddi{i}) (\nbth-1)}{(\rateCas{i} + \rateCasAddi{i})\nbth + (\rateRead{i} + \rateReadAddi{i})}  \]

\subsubsection{Stall Time}

Finally, packing has a potential to increase the ratio of time that
the cacheline (that \node{i} is mapped to) is blocked due to \castxts
executions. We simply update the process by involving the additive
factor:

\[ \expec{\cassta{i}} = (\rateCas{i} + \rateCasAddi{i})(\nbth-1)\latcas\frac{\latcas}{2} \]

\subsubsection{Experiments}

In Figures~\ref{fig:llpack} and~\ref{fig:htpack}, the results are depicted
for configurations with padding (dashed lines), packing(dots) and our
packing based estimations(lines), for the linked list and hash
table (nodes for tree and skiplist is too large to be
packed in a single cacheline or already packed). The key selection is done
with the uniform distribution.
For almost every case, we observe that the packing increases the performance
and the performance do not degrade due to the false sharing, even when the update rate is high.
The stall time (\expec{\cassta{i}}) often is not significant and the
invalidation recovery (\expec{\casrec{i}}) dominates the performance when there are
update operations. As an observation, the latency induced by this factor do
not increase with packing, presumably because:

\[ \frac{(\rateCas{i} +\rateCasAddi{i}) (\nbth-1)}{(\rateCas{i} + \rateCasAddi{i})\nbth + (\rateRead{i} + \rateReadAddi{i})} \approx \frac{\rateCas{i}(\nbth-1)}{\rateCas{i}\nbth + \rateRead{i}} \]

This might explain us the reason why the false sharing do not degrade the performance,
as opposed to one might expect. However, the cache and page misses influence the
performance positively, as expected.

Our estimations show that these effects are captured by our framework. We observe a slight
increase in almost all the curves that is coupled with a slight increase in our estimations,
due to the reduced capacity cache misses.

\begin{figure}[h!]
\begin{center}
\includegraphics[width=\textwidth]{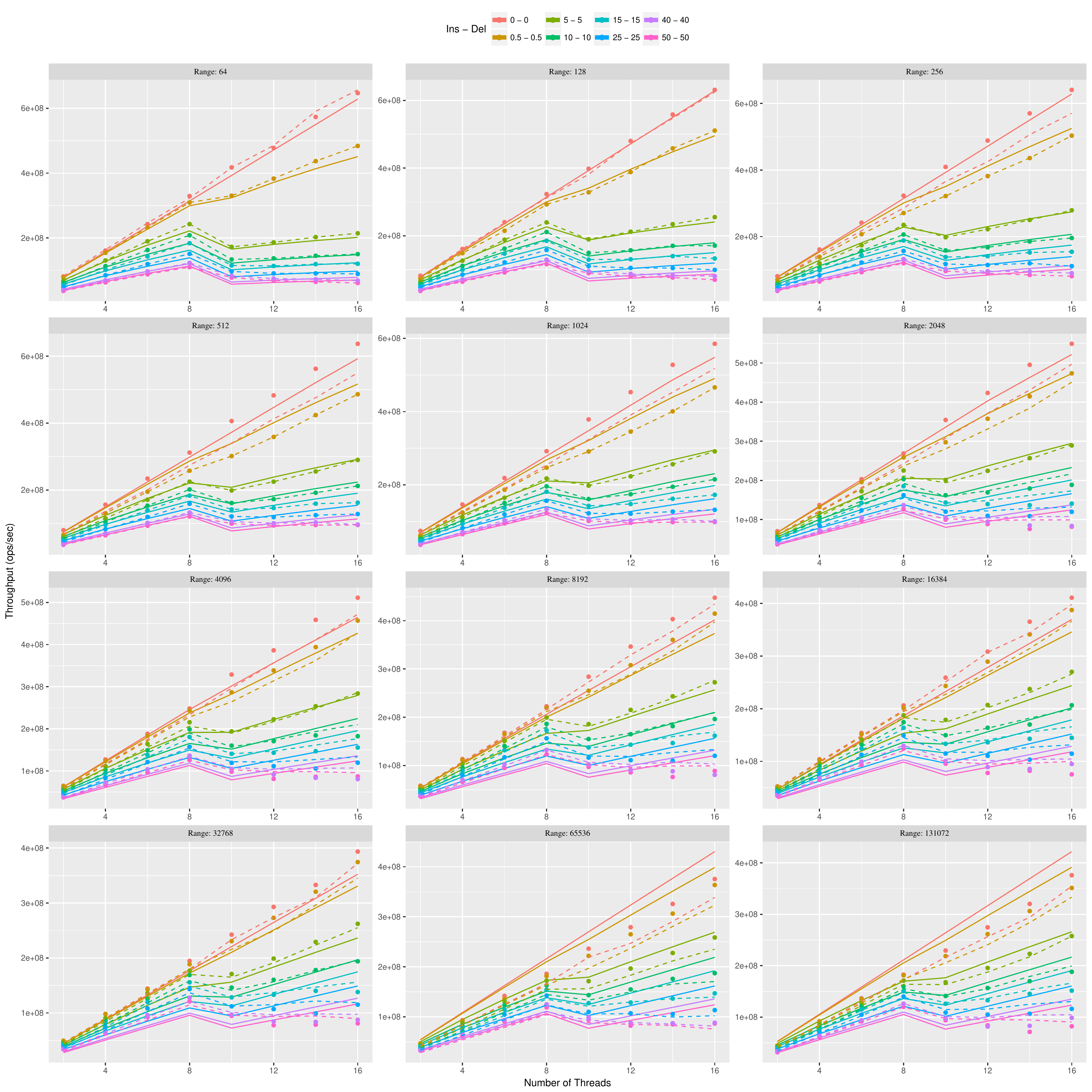}
\end{center}
\caption{Packed nodes for Hash Table, with load factor=2\label{fig:htpack}}
\end{figure}

\begin{figure}[h!]
\begin{center}
\includegraphics[width=\textwidth]{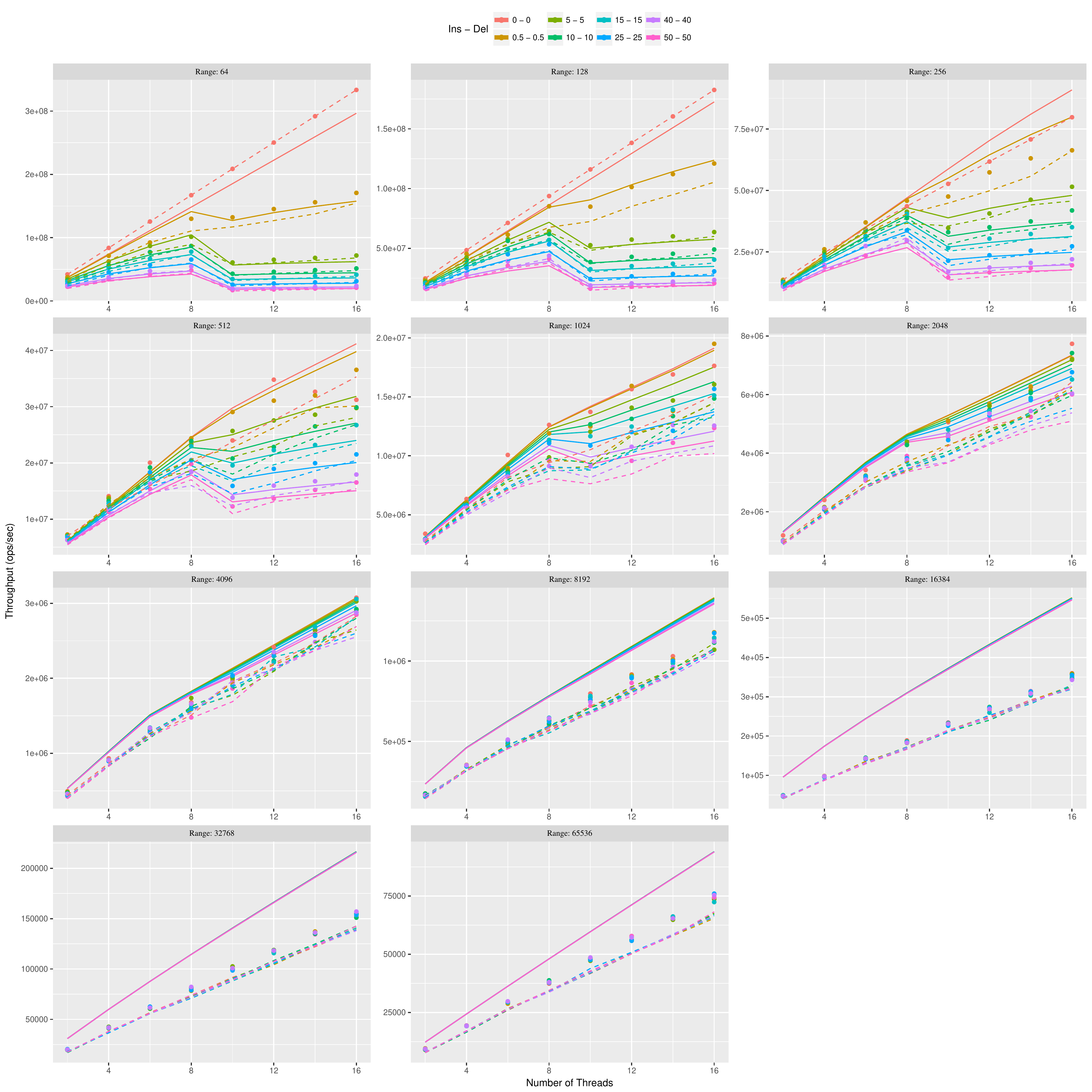}
\end{center}
\caption{Packed nodes for Linked List\label{fig:llpack}}
\end{figure}

}

\section{Conclusion\pp\vsme}

In this paper, we have modelled and analysed the performance of \sdss
under a stationary and memoryless access pattern. 
We have distinguished two types of events that occur in the \sds nodes and 
have modelled the arrival of events with Poisson processes. The properties of the
Poisson process allowed us to consider the thread-wise and system-wise interleaving of
events which are crucial for the estimation of the throughput. For the validation, we 
have used several fundemental \lf \sdss.

As a future work, it would be of interest to study to which extent the
application workload can be distorted while giving satisfactory
results. Putting aside the non-memoryless access patterns, the
non-stationary workloads such as bursty access patterns, could be
covered by splitting the time interval into alternating phases and
assuming a stationary behaviour for each phase.  Furthermore, we
foresee that the framework can capture the performance of lock-based
\sdss and also can be exploited to predict the energy efficiency of
the concurrent \sdss.

\rr{\newpage}
\bibliographystyle{plainurl}
\bibliography{./bibliography/longhead,./bibliography/biblio}

\end{document}